\documentclass[12pt]{article}
\usepackage{Sloczynski_style}

\renewcommand{\thesection}{\Roman{section}}
\renewcommand{\thesubsection}{\Alph{subsection}}

\renewcommand\footnotemark{}

\begin{document}
\doublespacing

\title{\textsc{\Large{Interpreting OLS Estimands When Treatment Effects Are Heterogeneous: Smaller Groups Get Larger Weights}}\thanks{\footnotesize{Tymon Słoczyński is an Assistant Professor at the Department of Economics and International Business School, Brandeis University. E-mail:~tslocz@brandeis.edu.}}
}
\author{\textsc{\large{Tymon Słoczyński}}\thanks{This version: May 18, 2020. This paper is based on portions of my previous working paper, \cite{Sloczynski2018}. I thank the editor and two anonymous referees for their helpful comments. I am very grateful to Alberto Abadie, Josh Goodman, Max Kasy, Pedro Sant'Anna, and Jeff Wooldridge for many comments and discussions. I also thank Arun Advani, Isaiah Andrews, Josh Angrist, Orley Ashenfelter, Richard Blundell, Stéphane Bonhomme, Carol Caetano, Marco Caliendo, Matias Cattaneo, Gary Chamberlain, Todd Elder, Alfonso Flores-Lagunes, Brigham Frandsen, Florian Gunsilius, Andreas Hagemann, James Heckman, Kei Hirano, Peter Hull, Macartan Humphreys, Guido Imbens, Krzysztof Karbownik, Shakeeb Khan, Toru Kitagawa, Pat Kline, Paweł Królikowski, Nicholas Longford, James MacKinnon, Łukasz Marć, Doug Miller, Michał Myck, Mateusz Myśliwski, Gary Solon, Jann Spiess, Michela Tincani, Alex Torgovitsky, Joanna Tyrowicz, Takuya Ura, Rudolf Winter-Ebmer, seminar participants at BC, Brandeis, Harvard--MIT, Holy Cross, IHS Vienna, Lehigh, MSU, Potsdam, SDU Odense, SGH, Temple, UCL, Upjohn, and WZB Berlin, and many conference participants for useful feedback. I thank Mark McAvoy for his excellent assistance in developing the R package hettreatreg that implements the results in this paper. I also thank David Card, Jochen Kluve, and Andrea Weber for providing me with supplementary data on the articles surveyed in \cite{CKW2018}. I acknowledge financial support from the National Science Centre (grant DEC-2012/05/N/HS4/00395), the Foundation for Polish Science (a ``Start'' scholarship), the ``Weź stypendium---dla rozwoju'' scholarship program, and the Theodore and Jane Norman Fund.}}
\date{}

\begin{titlepage}
\maketitle
\renewcommand{\abstractname}{}
\vspace{-40pt}
\begin{abstract}
\begin{small}
\noindent
Applied work often studies the effect of a binary variable (``treatment'') using linear models with additive effects. I study the interpretation of the OLS estimands in such models when treatment effects are heterogeneous. I show that the treatment coefficient is a convex combination of two parameters, which under certain conditions can be interpreted as the average treatment effects on the treated and untreated. The weights on these parameters are inversely related to the proportion of observations in each group. Reliance on these implicit weights can have serious consequences for applied work, as I illustrate with two well-known applications. I develop simple diagnostic tools that empirical researchers can use to avoid potential biases. Software for implementing these methods is available in R and Stata. In an important special case, my diagnostics only require the knowledge of the proportion of treated units.
\end{small}
\end{abstract}
\thispagestyle{empty}
\end{titlepage}

\setcounter{page}{2}

\setlength\abovedisplayskip{5pt}
\setlength\belowdisplayskip{5pt}
\setlength\abovedisplayshortskip{0pt}
\setlength\belowdisplayshortskip{5pt}

\section{Introduction}
\label{sec:intro}

Many applied researchers study the effect of a binary variable (``treatment'') on the expected value of an outcome of interest, holding fixed a vector of control variables. As noted by \cite{Imbens2015}, despite the availability of a large number of semi- and nonparametric estimators for average treatment effects, applied researchers often continue to use conventional regression methods. In particular, numerous studies use ordinary least squares (OLS) to estimate
\begin{equation}
\label{ols}
y = \alpha + \tau d + X \beta + u,
\end{equation}
where $y$ denotes the outcome, $d$ denotes the treatment, and $X$ denotes the row vector of control variables, $\left( x_1, \ldots, x_K \right)$. Usually, $\tau$ is interpreted as the average treatment effect (ATE)\@. This estimation strategy is used in many influential papers in economics \citep[e.g.,][]{VV2012, AGN2013, AEFLM2016}, as well as in other disciplines.

The great appeal of the model in (\ref{ols}) comes from its simplicity \citep{AP2009}. At the same time, however, a large body of evidence demonstrates the importance of heterogeneity in effects \citep[see, e.g.,][]{Heckman2001, BGH2006}, which is explicitly ruled out by this same model. In this paper I contribute to the recent literature on interpreting $\tau$, the OLS estimand, when treatment effects are heterogeneous \citep{Angrist1998, Humphreys2009, AS2016}. I demonstrate that $\tau$ is a convex combination of two parameters, which under certain conditions can be interpreted as the average treatment effects on the treated (ATT) and untreated (ATU)\@. Surprisingly, the weight that is placed by OLS on the average effect for each group is inversely related to the proportion of observations in this group. The more units are treated, the less weight is placed on ATT\@. One interpretation of this result is that OLS estimation of the model in (\ref{ols}) is generally inappropriate when treatment effects are heterogeneous.

It is also possible, however, to present a more pragmatic view of my main result. I derive a number of corollaries of this result which suggest several diagnostic methods that I recommend to applied researchers. These diagnostics are applicable whenever the researcher is: \textit{(i)} studying the effects of a binary treatment, \textit{(ii)} using OLS, and \textit{(iii)} unwilling to maintain that ATT is exactly equal to ATU\@. Typically, such a homogeneity assumption would be undesirably strong, because those choosing or chosen for treatment may have unusually high or low returns from that treatment, which would directly contradict the equality of ATT and ATU\@.

In deriving my diagnostics, I assume that the researcher is ultimately interested in ATE, ATT, or both, and that she wishes to estimate the model in (\ref{ols}) using OLS but is concerned about treatment effect heterogeneity. In this case, my diagnostics are able to detect deviations of the OLS weights from the pattern that would be necessary to consistently estimate a given parameter. These diagnostics are easy to implement and interpret; they are bounded between zero and one in absolute value and they give the proportion of the difference between ATU and ATT (or between ATT and ATU) that contributes to bias. Thus, if a given diagnostic is close to zero, OLS is likely a reasonable choice; but if a diagnostic is far from zero, other methods should be used.

In an important special case, these diagnostics become particularly simple and immediate to report. If we wish to estimate ATT, this ``rule of thumb'' variant of my diagnostic is equal to the proportion of treated units, $\pr \left( d=1 \right)$; if our goal is to estimate ATE, the diagnostic is equal to $2 \cdot \pr \left( d=1 \right) - 1$, twice the deviation of $\pr \left( d=1 \right)$ from 50\%. In short, OLS is expected to provide a reasonable approximation to ATE if both groups, treated and untreated, are of similar size. If we wish to estimate ATT, it is necessary that the proportion of treated units is very small.

It follows that OLS might often be substantially biased for ATE, ATT, or both. How common are these biases in practice? In a subset of 37 estimates from \cite{CKW2018}, a recent survey of evaluations of active labor market programs, the mean proportion of treated units is 17.7\%.\footnote{This sample is restricted to studies that \cite{CKW2018} coded as ``selection on observables'' and ``regression.''} Using the ``rule of thumb'' variants of my diagnostics, I establish that on average the difference between the OLS estimand and ATE is expected to correspond to 64.6\% of the difference between ATT and ATU\@. Similarly, the expected difference between OLS and ATT is on average equal to 17.7\% of the difference between ATU and ATT\@. In other words, these biases might often be large.

The remainder of the paper is organized as follows. Section \ref{sec:theory} presents a leading example and the main theoretical results. Section \ref{sec:empirical} discusses two empirical applications. In a study of the effects of a training program \citep{LaLonde1986}, OLS estimates are very similar to $\widehat{\mathrm{ATT}}$\@. On the other hand, in a study of the effects of cash transfers \citep{AEFLM2016}, OLS estimates are similar to $\widehat{\mathrm{ATU}}$\@. Section \ref{sec:conclusion} concludes. Proofs and several extensions are provided in the online appendices. The main results are implemented in newly developed R and Stata packages, \texttt{\small{hettreatreg}}.

\section{A Weighted Average Interpretation of OLS}
\label{sec:theory}

\subsection{Leading Example}
\label{sec:example}

To illustrate the problem with OLS weights, consider the classic example of the National Supported Work (NSW) program. Because this program originally involved a social experiment, the difference in mean outcomes between the treated and control units provides an unbiased estimate of the effect of treatment. \cite{LaLonde1986} studies the performance of various estimators at reproducing this experimental benchmark when the experimental controls are replaced by an artificial comparison group from the Current Population Survey (CPS) or the Panel Study of Income Dynamics (PSID)\@. \cite{AP2009} reanalyze the NSW--CPS data and conclude that OLS estimates of the effect of NSW program on earnings in 1978 are similar to the experimental benchmark of \$1,794.\footnote{Subsequently to \cite{LaLonde1986}, these data were studied by \cite{DW1999}, \cite{ST2005}, and many others. \cite{AP2009} analyze the subsample of the experimental treated units constructed by \cite{DW1999}, combined with ``CPS-1'' or ``CPS-3,'' i.e.~two of the nonexperimental comparison groups from CPS, constructed by \cite{LaLonde1986}. In this replication, I focus on ``CPS-1.''} In particular, their richest specification delivers an estimate of \$794. As I will show, this conclusion is driven by the small proportion of treated units in these data.

In this example, ATT and ATU are likely to be substantially different. This is because the treated group, unlike the CPS comparison (untreated) group, was highly economically disadvantaged. It is plausible that ATU might be zero or, due to the opportunity cost of program participation, even negative. Also, only 1.1\% of the sample was treated, so ATE and ATU will be similar.

To demonstrate this, I modify the model in (\ref{ols}) to include all interactions between $d$ and $X$\@. Estimation of this expanded model, again using OLS, allows us to separately compute $\widehat{\mathrm{ATE}}$, $\widehat{\mathrm{ATT}}$, and $\widehat{\mathrm{ATU}}$. This method is usually referred to as ``regression adjustment'' \citep{Wooldridge2010} or ``Oaxaca--Blinder'' \citep{Kline2011, GP2018}. Using the control variables that deliver the estimate of \$794, we obtain $\widehat{\mathrm{ATE}} = -\$4 \mathrm{,} 930$, $\widehat{\mathrm{ATT}} = \$796$, and $\widehat{\mathrm{ATU}} = -\$4 \mathrm{,} 996$. It turns out that, since $\widehat{\mathrm{ATE}}$ and $\widehat{\mathrm{ATU}}$ are indeed negative, the OLS estimate and $\widehat{\mathrm{ATE}}$ have different signs. Moreover, if we represent the OLS estimate as a weighted average of $\widehat{\mathrm{ATT}}$ and $\widehat{\mathrm{ATU}}$ with weights that sum to unity, we can write $\$794 = \hat{w}_{ATT} \cdot \$796 + \left( 1-\hat{w}_{ATT} \right) \cdot \left( -\$4 \mathrm{,} 996 \right)$, where $\hat{w}_{ATT}$ is the weight on $\widehat{\mathrm{ATT}}$. Solving for $\hat{w}_{ATT}$ yields $\hat{w}_{ATT} = 99.96\%$. In other words, the hypothetical OLS weight on the effect on the treated is similar to the proportion of untreated units, 98.9\%.

This ``weight reversal'' is not a coincidence. As I demonstrate below, the intuition from this example holds more generally, even though the OLS estimand is not necessarily a convex combination of two parameters from a procedure that controls for the full vector $X$\@.

\subsection{Main Result}

This section presents my main result, which focuses on the algebra of OLS and ``descriptive'' estimands that I define below. A causal interpretation of OLS also requires introducing the notion of potential outcomes as well as certain conditions that I discuss in section \ref{sec:theory}\ref{sec:causal}, including an ignorability assumption. However, this is not needed for my main result.

If $\lp \left( \cdot \mid \cdot \right)$ denotes the linear projection, we are interested in the interpretation of $\tau$ in the linear projection of $y$ on $d$ and $X$,
\begin{equation}
\label{lp_y}
\lp \left( y \mid 1, d, X \right) = \alpha + \tau d + X \beta,
\end{equation}
when this linear projection does not correspond to the (structural) conditional mean. Let
\begin{equation}
\rho = \pr \left( d=1 \right)
\end{equation}
be the unconditional probability of treatment and let
\begin{equation}
\label{lp_d}
p \left( X \right) = \lp \left( d \mid 1, X \right) = \alpha_p + X \beta_p
\end{equation}
be the ``propensity score'' from the linear probability model or, equivalently, the best linear approximation to the true propensity score. Generally, the specification in (\ref{lp_y}) and (\ref{lp_d}) can be arbitrarily flexible, so this approximation can be made very accurate; in fact, we can think of equation (\ref{lp_y}) as partially linear, where we may include powers and cross-products of original control variables.

After defining $p \left( X \right)$, it is helpful to introduce two linear projections of $y$ on $p \left( X \right)$, separately for $d=1$ and $d=0$, namely
\begin{equation}
\label{lp_y1}
\lp \left[ y \mid 1, p \left( X \right), d=1 \right] = \alpha_1 + \gamma_1 \cdot p \left( X \right)
\end{equation}
and also
\begin{equation}
\label{lp_y0}
\lp \left[ y \mid 1, p \left( X \right), d=0 \right] = \alpha_0 + \gamma_0 \cdot p \left( X \right).
\end{equation}
Note that equations (\ref{lp_d}), (\ref{lp_y1}), and (\ref{lp_y0}) are definitional. It is sufficient for my main result that the linear projections introduced so far exist and are unique.

\begin{assumption}
\label{ass:ols}
(i) $\e ( y^2 )$ and $\e ( \| X \| ^2 )$ are finite. (ii) The covariance matrix of $\left( d,X \right)$ is nonsingular.
\end{assumption}

\begin{assumption}
\label{ass:px}
$\var \left[ p \left( X \right) \mid d=1 \right]$ and $\var \left[ p \left( X \right) \mid d=0 \right]$ are nonzero, where $\var \left( \cdot \mid \cdot \right)$ denotes the conditional variance (with respect to $\e \left[ p \left( X \right) \mid d=j \right]$, $j=0,1$).
\end{assumption}

\noindent
Assumption \ref{ass:ols} guarantees the existence and uniqueness of the linear projections in (\ref{lp_y}) and (\ref{lp_d}). Similarly, Assumption \ref{ass:px} ensures that the linear projections in (\ref{lp_y1}) and (\ref{lp_y0}) exist and are unique.\footnote{Both assumptions are generally innocuous, although Assumption \ref{ass:px} rules out a small number of interesting applications, such as regression adjustments in Bernoulli trials and completely randomized experiments. In these cases, however, OLS is consistent for the average treatment effect under general conditions \citep{IR2015}.}

The next step is to use the linear projections in (\ref{lp_y1}) and (\ref{lp_y0}) to define the average partial linear effect of $d$ as
\begin{equation}
\label{tau_ape}
\tau_{APLE} = \left( \alpha_1 - \alpha_0 \right) + \left( \gamma_1 - \gamma_0 \right) \cdot \e \left[ p \left( X \right) \right]
\end{equation}
as well as the average partial linear effect of $d$ on group $j$ ($j=0,1$) as
\begin{equation}
\label{tau_apej}
\tau_{APLE, j} = \left( \alpha_1 - \alpha_0 \right) + \left( \gamma_1 - \gamma_0 \right) \cdot \e \left[ p \left( X \right) \mid d=j \right].
\end{equation}
These estimands are well defined under Assumptions \ref{ass:ols} and \ref{ass:px}, and have a causal interpretation under additional assumptions, as discussed in section \ref{sec:theory}\ref{sec:causal} below.\footnote{Moreover, $\tau_{APLE}$ is similar to the ``average regression coefficient'' or ``average slope coefficient'' in \cite{GP2018}, which is also a descriptive estimand in the sense of \cite{AAIW2020}.} When the linear projections in equations (\ref{lp_y1}) and (\ref{lp_y0}) represent the conditional mean of $y$, the average partial linear effects of $d$ overlap with its average partial effects. It should be stressed, however, that Theorem \ref{the:ols}, the main result of this paper, is more general and only requires Assumptions \ref{ass:ols} and \ref{ass:px}.

\begin{theorem}[Weighted Average Interpretation of OLS]
\label{the:ols}
Under Assumptions \ref{ass:ols} and \ref{ass:px},
\begin{eqnarray}
\tau &=& w_1 \cdot \tau_{APLE, 1} + w_0 \cdot \tau_{APLE, 0},
\nonumber
\end{eqnarray}
where $w_1 = \frac{\left( 1 - \rho \right) \cdot \var \left[ p \left( X \right) \mid d=0 \right]}{\rho \cdot \var \left[ p \left( X \right) \mid d=1 \right] + \left( 1 - \rho \right) \cdot \var \left[ p \left( X \right) \mid d=0 \right]}$ and $w_0 = 1 - w_1 = \frac{\rho \cdot \var \left[ p \left( X \right) \mid d=1 \right]}{\rho \cdot \var \left[ p \left( X \right) \mid d=1 \right] + \left( 1 - \rho \right) \cdot \var \left[ p \left( X \right) \mid d=0 \right]}$.
\end{theorem}
\begin{proof}
See online appendix \ref{app:proofs}\@.
\end{proof}

\noindent
Theorem \ref{the:ols} shows that $\tau$, the OLS estimand, is a convex combination of $\tau_{APLE, 1}$ and $\tau_{APLE, 0}$. The definition of $\tau_{APLE, j}$ makes it clear that $\tau$ is equivalent to the outcome of a particular three-step procedure. In the first step, we obtain $p \left( X \right)$, i.e.~the ``propensity score.'' Next, in the second step, we obtain $\tau_{APLE, 1}$ and $\tau_{APLE, 0}$, as in (\ref{tau_apej}), from two linear projections of $y$ on $p \left( X \right)$, separately for $d=1$ and $d=0$. This is analogous to the ``regression adjustment'' procedure in section \ref{sec:theory}\ref{sec:example}, although now we control for $p \left( X \right)$ rather than the full vector $X$\@. Finally, in the third step, we calculate a weighted average of $\tau_{APLE, 1}$ and $\tau_{APLE, 0}$. The weight on $\tau_{APLE, 1}$, $w_1$, is decreasing in $\frac{\var \left[ p \left( X \right) \mid d=1 \right]}{\var \left[ p \left( X \right) \mid d=0 \right]}$ and $\rho$ and the weight on $\tau_{APLE, 0}$, $w_0$, is increasing in $\frac{\var \left[ p \left( X \right) \mid d=1 \right]}{\var \left[ p \left( X \right) \mid d=0 \right]}$ and $\rho$.\footnote{A formal proof that the relationship between $\rho$ and $w_1$ ($w_0$) is indeed always negative (positive) is provided in online appendix \ref{app:monotonic}\@. This proof additionally assumes that the conditional mean of $d$ is linear in $X$\@.} This is clearly undesirable, since $\tau_{APLE} = \rho \cdot \tau_{APLE, 1} + \left( 1 - \rho \right) \cdot \tau_{APLE, 0}$.

This weighting scheme is also surprising: the more units belong to group $j$, the less weight is placed on $\tau_{APLE, j}$, i.e.~the effect \textit{for this group}. There are several ways to provide intuition for this result. One is provided in the next section. Another intuition follows from an alternative proof of Theorem \ref{the:ols}, which is provided with discussion in online appendix \ref{app:resid}\@. It parallels the intuition in \cite{Angrist1998} and \cite{AP2009} that OLS gives more weight to treatment effects that are better estimated in finite samples.\footnote{This proof uses a result from \cite{Deaton1997} and \cite{SHW2015} as a lemma. The main proof of Theorem \ref{the:ols} uses a result on decomposition methods from \cite{EGH2010}. See online appendix \ref{app:proofs} for more details.}

\subsection{Causal Interpretation}
\label{sec:causal}

The fact that Theorem \ref{the:ols} only requires the existence and uniqueness of several linear projections makes this result very general. On the other hand, one concern about this result might be that $\tau_{APLE, 1}$ and $\tau_{APLE, 0}$ do not necessarily correspond to the usual (causal) objects of interest. To define these objects, we need two potential outcomes, $y(1)$ and $y(0)$, only one of which is observed for each unit, $y = y(d) = y(1) \cdot d + y(0) \cdot \left( 1-d \right)$. The parameters of interest, ATE, ATT, and ATU, are defined as $\tau_{ATE} = \e \left[ y(1) - y(0) \right]$, $\tau_{ATT} = \e \left[ y(1) - y(0) \mid d=1 \right]$, and $\tau_{ATU} = \e \left[ y(1) - y(0) \mid d=0 \right]$. A causal interpretation of OLS also entails the following assumptions.

\begin{assumption}[Ignorability in Mean]
\label{ass:uncon}
(i) $\e \left[ y(1) \mid X,d \right] = \e \left[ y(1) \mid X \right]$; and (ii) $\e \left[ y(0) \mid X,d \right] = \e \left[ y(0) \mid X \right]$.
\end{assumption}

\begin{assumption}
\label{ass:lin}
(i) $\e \left[ y(1) \mid X \right] = \alpha_1 + \gamma_1 \cdot p \left( X \right)$; and (ii) $\e \left[ y(0) \mid X \right] = \alpha_0 + \gamma_0 \cdot p \left( X \right)$.
\end{assumption}

\noindent
Assumptions \ref{ass:uncon} and \ref{ass:lin} ensure that $\tau$ admits a causal interpretation. Assumption \ref{ass:uncon} is standard in the program evaluation literature \citep{Wooldridge2010}. Assumption \ref{ass:lin} is not commonly used. Sufficient for this assumption, but not necessary, is that the conditional mean of $d$ is linear in $X$ and the conditional means of $y(1)$ and $y(0)$ are linear in the true propensity score, which is now equal to $p \left( X \right)$. Linearity of $\e \left( d \mid X \right)$ is assumed in \cite{AS2016} and \cite{AAIW2020}. This assumption is not necessarily strong, since $X$ might include powers and cross-products of original control variables. It is also satisfied automatically in saturated models, as in \cite{Angrist1998} and \cite{Humphreys2009}. The linearity assumption for $\e \left[ y(1) \mid p \left( X \right) \right]$ and $\e \left[ y(0) \mid p \left( X \right) \right]$ dates back to \cite{RR1983} but is restrictive. See also \cite{IW2009} and \cite{Wooldridge2010} for a discussion.

\begin{corollary}[Causal Interpretation of OLS]
\label{cor:causal}
Under Assumptions \ref{ass:ols}, \ref{ass:px}, \ref{ass:uncon}, and \ref{ass:lin},
\begin{eqnarray}
\tau & = & w_1 \cdot \tau_{ATT} + w_0 \cdot \tau_{ATU}.
\nonumber
\end{eqnarray}
\end{corollary}
\begin{proof}
Assumption \ref{ass:uncon} implies that $\e \left[ y(1)-y(0) \mid X \right] = \e \left( y \mid X,~d=1 \right) - \e \left( y \mid X,~d=0 \right)$. Then, Assumption \ref{ass:lin} implies that $\e \left[ y(1)-y(0) \mid X \right] = \left( \alpha_1 - \alpha_0 \right) + \left( \gamma_1 - \gamma_0 \right) \cdot p \left( X \right)$, which in turn implies that $\tau_{ATT} = \tau_{APLE, 1}$ and $\tau_{ATU} = \tau_{APLE, 0}$. This, together with Theorem \ref{the:ols}, completes the proof.
\end{proof}

\noindent
Corollary \ref{cor:causal} states that, under Assumptions \ref{ass:ols}, \ref{ass:px}, \ref{ass:uncon}, and \ref{ass:lin}, the OLS weights from Theorem \ref{the:ols} apply to the causal objects of interest, $\tau_{ATT}$ and $\tau_{ATU}$. Hence, $\tau$ has a causal interpretation. The greater the proportion of treated units, the smaller is the OLS weight on $\tau_{ATT}$. Again, this is undesirable, since $\tau_{ATE} = \rho \cdot \tau_{ATT} + \left( 1 - \rho \right) \cdot \tau_{ATU}$.

To aid intuition for this surprising result, recall that an important motivation for using the model in (\ref{ols}) and OLS is that the linear projection of $y$ on $d$ and $X$ provides the best linear predictor of $y$ given $d$ and $X$ \citep{AP2009}. However, if our goal is to conduct causal inference, then this is not, in fact, a good reason to use this method. Ordinary least squares is ``best'' in predicting actual outcomes but causal inference is about predicting missing outcomes, defined as $y_m = y(1) \cdot \left( 1-d \right) + y(0) \cdot d$. In other words, the OLS weights are optimal for predicting ``what is.'' Instead, we are interested in predicting ``what would be'' if treatment were assigned differently.

Intuition suggests that if our goal were to predict ``what is'' and, without loss of generality, group one were substantially larger than group zero, we would like to place a large weight on the linear projection coefficients of group one ($\alpha_1$ and $\gamma_1$), because these coefficients can be used to predict actual outcomes of this group. As noted by \cite{Deaton1997} and \cite{SHW2015}, the OLS weights are consistent with this idea. Indeed, Theorem \ref{the:ols} also implies that
\begin{equation}
\label{tau_decomp}
\tau = \left[ \e \left( y \mid d=1 \right) - \e \left( y \mid d=0 \right) \right] - \left( w_0 \gamma_1 + w_1 \gamma_0 \right) \cdot \left\lbrace \e \left[ p \left( X \right) \mid d=1 \right] - \e \left[ p \left( X \right) \mid d=0 \right] \right\rbrace.
\end{equation}
Namely, the OLS estimand is equal to the simple difference in means of $y$ plus an adjustment term that depends on the difference in means of $p \left( X \right)$ and a weighted average of $\gamma_1$ and $\gamma_0$. When group one is ``large,'' $w_0$, the weight on $\gamma_1$, is large as well.

Conversely, if group one is ``large'' but our goal is to predict missing outcomes, we need to place a large weight on $\alpha_0$ and $\gamma_0$, because these coefficients can be used to predict counterfactual outcomes of group one. To see this point, note that it follows from the discussion in \cite{IW2009} that when the conditional means of $y(1)$ and $y(0)$ are linear in $X$, we can write
\begin{equation}
\label{ate_decomp}
\tau_{ATE} = \left[ \e \left( y \mid d=1 \right) - \e \left( y \mid d=0 \right) \right] - \left[ \left( 1 - \rho \right) \beta_1 + \rho \beta_0 \right] \cdot \left[ \e \left( X \mid d=1 \right) - \e \left( X \mid d=0 \right) \right],
\end{equation}
where $\beta_1$ and $\beta_0$ are the coefficients on $X$ in the conditional means of $y(1)$ and $y(0)$, respectively. Equations (\ref{tau_decomp}) and (\ref{ate_decomp}) reiterate the point of Corollary \ref{cor:causal} that $\tau$ and $\tau_{ATE}$ have a very similar structure but they differ substantially in how they assign weights. Indeed, in the case of $\tau_{ATE}$, when group one is ``large,'' the weight on $\beta_1$ is small, the opposite of what we have seen for OLS\@.\footnote{Note that the (infeasible) linear projection of the missing outcome, $y_m$, on $d$ and $X$ would solve our problem of ``weight reversal.'' The weights on $\tau_{ATT}$ and $\tau_{ATU}$ would still be different than $\rho$ and $1 - \rho$ if $\var \left[ p \left( X \right) \mid d=1 \right]$ and $\var \left[ p \left( X \right) \mid d=0 \right]$ were different; but, at least, the weight on $\tau_{ATT}$ ($\tau_{ATU}$) would be increasing (decreasing) in $\rho$.}

\subsection{Implications of Theorem \ref{the:ols}}
\label{sec:corollaries}

There are several practical implications of my main result. Throughout this section, I assume that the researcher is interested in estimating $\tau_{ATE}$, $\tau_{ATT}$, or both, and that she wishes to use OLS to estimate the model in (\ref{ols}) but is concerned about the implications of Theorem \ref{the:ols} and Corollary \ref{cor:causal}. In Corollaries \ref{cor:biasate} and \ref{cor:biasatt}, I show how to decompose the difference between $\tau$ and $\tau_{ATE}$ or $\tau$ and $\tau_{ATT}$ into components attributable to \textit{(i)} the difference between $\tau_{APLE, 1}$ and $\tau_{ATT}$, \textit{(ii)} the difference between $\tau_{APLE, 0}$ and $\tau_{ATU}$ (jointly referred to as ``bias from nonlinearity''), and \textit{(iii)} the OLS weights on $\tau_{ATT}$ and $\tau_{ATU}$ (``bias from heterogeneity'').\footnote{Because ``bias from nonlinearity'' arises when Assumptions \ref{ass:uncon} and/or \ref{ass:lin} are violated, it might be more accurate to refer to this component as ``bias from endogeneity and nonlinearity.'' Yet, I use the former term for brevity.} Because this paper generally focuses on what I now term ``bias from heterogeneity,'' my discussion below is restricted to this source of bias, which is equivalent to implicitly making Assumptions \ref{ass:uncon} and \ref{ass:lin}.

\begin{corollary}
\label{cor:biasate}
Under Assumptions \ref{ass:ols} and \ref{ass:px},
\begin{eqnarray}
\tau - \tau_{ATE} & = & \underbrace{w_0 \cdot \left( \tau_{APLE, 0} - \tau_{ATU} \right) + w_1 \cdot \left( \tau_{APLE, 1} - \tau_{ATT} \right)}_{\textrm{bias from nonlinearity}} \; + \; \underbrace{\delta \cdot \left( \tau_{ATU} - \tau_{ATT} \right)}_{\textrm{bias from heterogeneity}},
\nonumber
\end{eqnarray}
where $\delta = \rho - w_1 = \frac{\rho ^2 \cdot \var \left[ p \left( X \right) \mid d=1 \right] - \left( 1 - \rho \right) ^2 \cdot \var \left[ p \left( X \right) \mid d=0 \right]}{\rho \cdot \var \left[ p \left( X \right) \mid d=1 \right] + \left( 1 - \rho \right) \cdot \var \left[ p \left( X \right) \mid d=0 \right]}$. Also, under Assumptions \ref{ass:ols}, \ref{ass:px}, \ref{ass:uncon}, and \ref{ass:lin},
\begin{eqnarray}
\tau - \tau_{ATE} & = & \delta \cdot \left( \tau_{ATU} - \tau_{ATT} \right).
\nonumber
\end{eqnarray}
\end{corollary}
\begin{corollary}
\label{cor:biasatt}
Under Assumptions \ref{ass:ols} and \ref{ass:px},
\begin{eqnarray}
\tau - \tau_{ATT} & = & \underbrace{w_0 \cdot \left( \tau_{APLE, 0} - \tau_{ATU} \right) + w_1 \cdot \left( \tau_{APLE, 1} - \tau_{ATT} \right)}_{\textrm{bias from nonlinearity}} \; + \; \underbrace{w_0 \cdot \left( \tau_{ATU} - \tau_{ATT} \right)}_{\textrm{bias from heterogeneity}}.
\nonumber
\end{eqnarray}
Also, under Assumptions \ref{ass:ols}, \ref{ass:px}, \ref{ass:uncon}, and \ref{ass:lin},
\begin{eqnarray}
\tau - \tau_{ATT} & = & w_0 \cdot \left( \tau_{ATU} - \tau_{ATT} \right).
\nonumber
\end{eqnarray}
\end{corollary}

\noindent
The proofs of Corollaries \ref{cor:biasate} and \ref{cor:biasatt} follow from simple algebra and are omitted. These results show that, regardless of whether we focus on $\tau_{ATE}$ or $\tau_{ATT}$, the bias from heterogeneity is equal to the product of a particular measure of heterogeneity, namely the difference between $\tau_{ATU}$ and $\tau_{ATT}$, and an additional parameter that is easy to estimate, $\delta$ for $\tau_{ATE}$ and $w_0$ for $\tau_{ATT}$. While $w_0$ is guaranteed to be positive under Assumptions \ref{ass:ols} and \ref{ass:px}, $\delta$ may be positive or negative. Both $w_0$ and $\delta$, however, are bounded between zero and one in absolute value. Thus, $w_0$ and $\vert \delta \vert$ can be interpreted as the percentage of our measure of heterogeneity, $\tau_{ATU} - \tau_{ATT}$, which contributes to bias.\footnote{To be precise, $\vert \delta \vert$ can be interpreted as the percentage of $\sgn( \delta ) \cdot \left( \tau_{ATU} - \tau_{ATT} \right)$ that contributes to bias when focusing on $\tau_{ATE}$. Both $\delta$ and $w_0$ also have an intuitive interpretation as the difference between \textit{(i)} the weight that we should place on $\tau_{ATT}$ when focusing on $\tau_{ATE}$ or $\tau_{ATT}$ and \textit{(ii)} the weight that OLS actually places on this parameter. Indeed, $\delta$ is equal to the difference between $\rho$ and $w_1$. Similarly, $w_0 = 1-w_1$.} It might be useful to report estimates of $w_0$ and $\delta$ in studies that use OLS to estimate the model in (\ref{ols}).

As an example, consider the empirical application in section \ref{sec:theory}\ref{sec:example}\@. In this case, $\hat{w}_0 = 0.017$ and $\hat{\delta} = -0.971$. The interpretation of these estimates is as follows: if our goal is to estimate $\tau_{ATT}$, using the model in (\ref{ols}) and OLS is expected to bias our estimates by only 1.7\% of the difference between $\tau_{ATU}$ and $\tau_{ATT}$. If instead we wanted to interpret $\tau$ as $\tau_{ATE}$, our estimates would be biased by an estimated 97.1\% of the difference between $\tau_{ATT}$ and $\tau_{ATU}$. Thus, in this application, it might perhaps be acceptable to interpret $\tau$ as $\tau_{ATT}$ but clearly not as $\tau_{ATE}$.

\begin{assumption}
\label{ass:var}
$\var \left[ p \left( X \right) \mid d=1 \right] = \var \left[ p \left( X \right) \mid d=0 \right]$.
\end{assumption}

\noindent
The calculation of $\delta$ and $w_0$ is further simplified under Assumption \ref{ass:var}. If we use $\delta^*$ and $w_0^*$ to denote the values of $\delta$ and $w_0$ in this special case, we can write $\delta^* = 2 \rho - 1$ and $w_0^* = \rho$. In this setting, the knowledge of $\delta$ and $w_0$ only requires information on $\rho$, the proportion of units with $d=1$. Of course, the special case where $\var \left[ p \left( X \right) \mid d=1 \right] = \var \left[ p \left( X \right) \mid d=0 \right]$ is hardly to be expected in practice. Still, $\delta^* = 2 \rho - 1$ and $w_0^* = \rho$ can potentially serve as a rule of thumb.

The practical implications of Assumption \ref{ass:var} are particularly clear when $\rho$ is close to 0\%, 50\%, or 100\%. When few units are treated, $\tau \simeq \tau_{ATT}$. When most of the units are treated, $\tau \simeq \tau_{ATU}$. Finally, when both groups are of similar size, $\tau \simeq \tau_{ATE}$. This can also be seen from Corollary \ref{cor:reverse}.

\begin{corollary}
\label{cor:reverse}
Under Assumptions \ref{ass:ols}, \ref{ass:px}, and \ref{ass:var},
\begin{eqnarray}
\tau &=& \left( 1 - \rho \right) \cdot \tau_{APLE, 1} + \rho \cdot \tau_{APLE, 0}.
\nonumber
\end{eqnarray}
Also, under Assumptions \ref{ass:ols}, \ref{ass:px}, \ref{ass:uncon}, \ref{ass:lin}, and \ref{ass:var},
\begin{eqnarray}
\tau &=& \left( 1 - \rho \right) \cdot \tau_{ATT} + \rho \cdot \tau_{ATU}.
\nonumber
\end{eqnarray}
\end{corollary}

\noindent
The proof follows immediately from simple algebra. Corollary \ref{cor:reverse} provides conditions under which OLS reverses the ``natural'' weights on $\tau_{APLE, 1}$ and $\tau_{APLE, 0}$ (or $\tau_{ATT}$ and $\tau_{ATU}$). Indeed, under Assumption \ref{ass:var}, $\tau$ is a convex combination of group-specific average effects, with ``reversed'' weights attached to these parameters. Namely, the proportion of units with $d=1$ is used to weight the average effect of $d$ on group zero, and vice versa.

The results in this section allow empirical researchers to interpret the OLS estimand when treatment effects are heterogeneous. Alternatively, it might be sensible to use any of the standard estimators for average treatment effects under ignorability, such as regression adjustment (see section \ref{sec:theory}\ref{sec:example}), weighting, matching, and various combinations of these approaches.\footnote{For recent reviews, see \cite{IW2009}, \cite{Wooldridge2010}, and \cite{AC2018}.} It might also help to estimate a model with homogeneous effects using weighted least squares (WLS)\@. Indeed, in online appendix \ref{app:wls}, I demonstrate that when we regress $y$ on $d$ and $p \left( X \right)$, with weights of $\frac{1 - \rho}{w_0}$ for units with $d=1$ and $\frac{\rho}{w_1}$ for units with $d=0$, the WLS estimand is equal to $\tau_{APLE}$. In practice, of course, $\tau_{APLE}$ can also be obtained directly from equation (\ref{tau_ape}).

\subsection{Related Work}

This section discusses the relationship between my main result and those in \cite{Angrist1998} and \cite{Humphreys2009}. These papers focus on saturated models with discrete covariates, in which the estimating equation includes an indicator for each combination of covariate values (``stratum''). In particular, \cite{Angrist1998} provides a representation of $\tau_n$ in
\begin{equation}
\lp \left( y \mid d, x_1, \ldots, x_{S} \right) = \tau_n d + \sum_{s=1}^{S} \beta_{n,s} x_s,
\label{angrist_reg}
\end{equation}
where $x_1, \ldots, x_S$ are stratum indicators. More precisely, \cite{Angrist1998} demonstrates that
\begin{equation}
\tau_n = \sum_{s=1}^S \frac{\pr \left( x_s=1 \right) \cdot \pr \left( d=1 \mid x_s=1 \right) \cdot \pr \left( d=0 \mid x_s=1 \right)}{\sum_{t=1}^S \pr \left( x_t=1 \right) \cdot \pr \left( d=1 \mid x_t=1 \right) \cdot \pr \left( d=0 \mid x_t=1 \right)} \cdot \tau_s,
\label{angrist_S}
\end{equation}
where $\tau_s = \e \left( y \mid d=1, x_s=1 \right) - \e \left( y \mid d=0, x_s=1 \right)$. In online appendix \ref{app:angrist}, I demonstrate that this result follows from Corollary \ref{cor:causal} when the model for $y$ is saturated.\footnote{Also, note that \cite{AS2016} show that this result in \cite{Angrist1998} is not specific to saturated models; instead, it is sufficient to assume that the model for $d$ is linear in $X$\@. My analysis in online appendix \ref{app:angrist} covers the results in both \cite{Angrist1998} and \cite{AS2016}.} At the same time, the interpretation of OLS in \cite{Angrist1998} is different from Theorem \ref{the:ols} and Corollary \ref{cor:causal}. On the one hand, unlike Corollary \ref{cor:causal} and \cite{Humphreys2009}, \cite{Angrist1998} does not restrict the relationship between $\tau_s$ and $\pr \left( d=1 \mid x_s=1 \right)$ in any way. On the other hand, Theorem \ref{the:ols} and Corollary \ref{cor:causal} make it arguably easier to identify whether in a given application the OLS estimand will be close to any of the parameters of interest (cf.~Corollaries \ref{cor:biasate} to \ref{cor:reverse}). In particular, \cite{Angrist1998} does not recover a pattern of ``weight reversal,'' which is discussed in detail in this paper.

Unlike \cite{Angrist1998}, \cite{Humphreys2009} does not derive a new representation of $\tau_n$, but instead presents further analysis of the result in equation (\ref{angrist_S}). In particular, \cite{Humphreys2009} notes that $\tau_n$ can take any value between $\min ( \tau_s )$ and $\max ( \tau_s )$. Then, he demonstrates that $\tau_n$ is also bounded by $\tau_{ATT}$ and $\tau_{ATU}$ if we restrict the relationship between $\tau_s$ and $\pr \left( d=1 \mid x_s=1 \right)$ to be monotonic. According to Corollary \ref{cor:causal}, $\tau$ is a convex combination of $\tau_{ATT}$ and $\tau_{ATU}$ if, among other things, both potential outcomes are linear in $p \left( X \right)$, which also implies a linear relationship between $\tau_s$ and $\pr \left( d=1 \mid x_s=1 \right)$ when the model for $y$ is saturated. Of course, this linearity assumption is stronger than the monotonicity assumption in \cite{Humphreys2009}. However, in return, we are able to derive a closed-form expression for $\tau$ in terms of $\tau_{ATT}$ and $\tau_{ATU}$, which is a major advantage over the earlier literature, such as \cite{Angrist1998} and \cite{Humphreys2009}.\footnote{\cite{Humphreys2009} also provides a brief informal remark that the OLS estimand, as represented in \cite{Angrist1998}, is similar to $\tau_{ATT}$ ($\tau_{ATU}$) if propensity scores are ``small'' (``large'') in \textit{every} stratum. This is a special case of the rule of thumb derived from Corollaries \ref{cor:biasatt} and \ref{cor:reverse}. My rule of thumb does not impose any such restrictions on the propensity score other than the requirement that the \textit{unconditional} probability of treatment is close to zero or one.}

\section{Empirical Applications}
\label{sec:empirical}

This section discusses two empirical illustrations of Theorem \ref{the:ols} and its corollaries.\footnote{In a follow-up paper, I apply these results in the study of racial gaps in test scores and wages \citep{Sloczynski_ILRR}.} In online appendices \ref{app:stata} and \ref{app:r}, I discuss the implementation of these results in Stata and R\@. Throughout the current section $\tau_{APLE}$, $\tau_{APLE, 1}$, and $\tau_{APLE, 0}$ are implicitly treated as equivalent to $\tau_{ATE}$, $\tau_{ATT}$, and $\tau_{ATU}$, respectively. Although this might be restrictive, I also demonstrate that in both applications sample analogues of $\tau_{APLE}$, $\tau_{APLE, 1}$, and $\tau_{APLE, 0}$, reported in the body of the paper, are similar to other estimates of $\tau_{ATE}$, $\tau_{ATT}$, and $\tau_{ATU}$, reported in online appendix \ref{app:robustness}\@.

\subsection{The Effects of a Training Program on Earnings}
\label{sec:nsw}

I first consider the example from section \ref{sec:theory}\ref{sec:example} in more detail. This replication of the study of the effects of NSW program in \cite{AP2009} constitutes an optimistic scenario for OLS\@. In this application, as I explained in section \ref{sec:theory}\ref{sec:example}, the effect for the treated group (ATT) is likely to be substantially larger than the effect for the CPS comparison group (ATU)\@. Moreover, since the experimental benchmark of \$1,794 corresponds to $\widehat{\mathrm{ATT}}$ and not to $\widehat{\mathrm{ATU}}$, the researcher should also focus on ATT\@. It turns out that my diagnostic for estimating ATT, $\hat{w}_0$, indicates that this parameter should approximately be recovered by OLS, even if treatment effects are heterogeneous.\footnote{It is well known that, in the NSW--CPS data, there is limited overlap in terms of covariate values between the treated and untreated units \citep[see, e.g.,][]{DW1999, ST2005}. Thus, it is important to note that my theoretical results in section \ref{sec:theory} do not impose the overlap assumption.}

\begin{table}[!tb]
\begin{adjustwidth}{-1in}{-1in}
\centering
\begin{threeparttable}
\caption{\normalsize{The Effects of a Training Program on Earnings\label{tab:nsw1}}}
\begin{normalsize}
\begin{tabular}{l >{\centering\arraybackslash}m{3.5cm} >{\centering\arraybackslash}m{3.5cm} >{\centering\arraybackslash}m{3.5cm} >{\centering\arraybackslash}m{3.5cm}}
\hline\hline
    \multicolumn{1}{c}{} & (1) & (2) & (3) & (4) \\
\hline
    \multicolumn{1}{c}{} & \multicolumn{4}{c}{Original estimates}  \\
\cline{2-5}
    \multicolumn{1}{l}{OLS} & --3,437*** & --78  & 623   & 794 \\
    \multicolumn{1}{c}{} & (612) & (596) & (610) & (619) \\
    \multicolumn{1}{l}{} & & & & \\
    \multicolumn{1}{c}{} & \multicolumn{4}{c}{Diagnostics}    \\
\cline{2-5}
    \multicolumn{1}{l}{$\hat{w}_0$} & 0.019 & 0.001 & 0.017 & 0.017 \\
    \multicolumn{1}{l}{$\hat{w}_0^* = \hat{\rho}$} & 0.011 & 0.011 & 0.011 & 0.011 \\
    \multicolumn{1}{l}{$\hat{\delta}$} & --0.970 & --0.987 & --0.971 & --0.971 \\
    \multicolumn{1}{l}{$\hat{\delta}^* = 2 \hat{\rho} - 1$} & --0.977 & --0.977 & --0.977 & --0.977 \\
    \multicolumn{1}{l}{} & & & & \\
    \multicolumn{1}{c}{} & \multicolumn{4}{c}{Decomposition}    \\
\cline{2-5}
    \multicolumn{1}{l}{$\widehat{\mathrm{ATT}}$} & --3,373*** & --69  & 754   & 928 \\
    \multicolumn{1}{l}{} & (620) & (595) & (619) & (630) \\
    \multicolumn{1}{l}{$\hat{w}_1$} & 0.981 & 0.999 & 0.983 & 0.983 \\
    \multicolumn{1}{l}{} & & & & \\
    \multicolumn{1}{l}{$\widehat{\mathrm{ATU}}$} & --6,753*** & --6,289** & --6,841*** & --6,840*** \\
    \multicolumn{1}{l}{} & (1,219) & (2,807) & (1,294) & (1,319) \\
    \multicolumn{1}{l}{$\hat{w}_0$} & 0.019 & 0.001 & 0.017 & 0.017 \\
    \multicolumn{1}{l}{} & & & & \\
    \multicolumn{1}{l}{$\widehat{\mathrm{ATE}}$} & --6,714*** & --6,218** & --6,754*** & --6,751*** \\
    \multicolumn{1}{l}{} & (1,206) & (2,777) & (1,281) & (1,305) \\
    \multicolumn{1}{l}{} & & & & \\
    \multicolumn{1}{l}{Demographic controls} & \checkmark & & \checkmark & \checkmark \\
    \multicolumn{1}{l}{Earnings in 1974} & & & & \checkmark \\
    \multicolumn{1}{l}{Earnings in 1975} & & \checkmark & \checkmark & \checkmark \\
    \multicolumn{1}{l}{} & & & & \\
    \multicolumn{1}{l}{$\hat{\rho} = \hat{\pr} \left( d=1 \right)$} & 0.011 & 0.011 & 0.011 & 0.011 \\
    \multicolumn{1}{l}{Observations} & 16,177 & 16,177 & 16,177 & 16,177 \\
\hline
\end{tabular}
\end{normalsize}
\begin{footnotesize}
\begin{tablenotes}[flushleft]
\item \textit{Notes:} The estimates in the top panel correspond to column 2 in Table 3.3.3 in \citet[][p.~89]{AP2009}. The dependent variable is earnings in 1978. Demographic controls include age, age squared, years of schooling, and indicators for married, high school dropout, black, and Hispanic. For treated individuals, earnings in 1974 correspond to real earnings in months 13--24 prior to randomization, which overlaps with calendar year 1974 for a number of individuals. Formulas for $w_0$, $w_1$, and $\delta$ are given in Theorem \ref{the:ols} and Corollary \ref{cor:biasate}. Following these results, $\mathrm{OLS} = \hat{w}_1 \cdot \widehat{\mathrm{ATT}} + \hat{w}_0 \cdot \widehat{\mathrm{ATU}}$\@. Estimates of ATE, ATT, and ATU are sample analogues of $\tau_{APLE}$, $\tau_{APLE, 1}$, and $\tau_{APLE, 0}$, respectively. Also, $\widehat{\mathrm{ATE}} = \hat{\rho} \cdot \widehat{\mathrm{ATT}} + \left( 1 - \hat{\rho} \right) \cdot \widehat{\mathrm{ATU}}$\@. Huber--White standard errors (OLS) and bootstrap standard errors ($\widehat{\mathrm{ATE}}$, $\widehat{\mathrm{ATT}}$, and $\widehat{\mathrm{ATU}}$) are in parentheses.
\item *Statistically significant at the 10\% level; **at the 5\% level; ***at the 1\% level.
\end{tablenotes}
\end{footnotesize}
\end{threeparttable}
\end{adjustwidth}
\end{table}

The top and middle panels of Table \ref{tab:nsw1} reproduce the estimates from \cite{AP2009} and report my diagnostics. The specification in column 4 was discussed in section \ref{sec:theory}\ref{sec:example}\@. It turns out that $\hat{w}_0$ is between 0.1\% and 1.9\% for all specifications; similarly, the ``rule of thumb'' value of this diagnostic, $\hat{w}_0^*$, is, as always, equal to the proportion of treated units (only 1.1\% in this sample). These results are very simple to interpret. Namely, as in section \ref{sec:theory}\ref{sec:corollaries}, we estimate that the difference between the OLS estimand and ATT is less than 2\% of the difference between ATU and ATT\@. In this case, it might indeed be sensible to rely on the OLS estimates of the effect of treatment.

The bottom panel of Table \ref{tab:nsw1} provides an application of Corollary \ref{cor:causal} to these results. In other words, the estimates from \cite{AP2009} are now decomposed into two components, $\widehat{\mathrm{ATT}}$ and $\widehat{\mathrm{ATU}}$\@. The difference between these estimates is substantial. In column 4, while the estimate of ATT is \$928, ATU is estimated to be --\$6,840. In other words, the OLS estimate of \$794, reported in \cite{AP2009} and discussed in section \ref{sec:theory}\ref{sec:example}, is actually a weighted average of these two estimates. The fact that it is close to \$928, and not to --\$6,840, is a consequence of the small proportion of treated units in this sample, 1.1\%. The weight on \$928, $\hat{w}_1$, is 98.3\% and the weight on --\$6,840, $\hat{w}_0$, is only 1.7\%.

We might expect that if the proportion of treated units was larger, the weight on $\widehat{\mathrm{ATT}}$ would be smaller and the ``performance'' of OLS in replicating the experimental benchmark would deteriorate. I confirm this conjecture in online appendix \ref{app:mhe} by quasi-discarding ``random'' subsamples of untreated units over a range of sample sizes. In particular, I reestimate the model in (\ref{ols}) using WLS, with weights of 1 for treated and $\frac{1}{k}$ for untreated units. Figures \ref{fig:nswwls1} to \ref{fig:nswwls4} show that in this application WLS estimates become more negative as $k$ increases. This is because larger values of $k$ correspond to greater proportions of untreated units being ``discarded,'' and hence \textit{larger} weights on $\widehat{\mathrm{ATU}}$, which is substantially more negative than $\widehat{\mathrm{ATT}}$\@.

Additional extensions of my analysis are also presented in online appendix \ref{app:mhe}\@. For each specification in Table \ref{tab:nsw1}, I provide both a linear and a nonparametric estimate of the conditional mean of the outcome given $p \left( X \right)$, separately for treated and untreated units (Figures \ref{fig:nswypx1} to \ref{fig:nswypx4})\@. A visual comparison of both estimates provides an informal test of Assumption \ref{ass:lin}, which is necessary for a causal interpretation of $\tau_{APLE}$, $\tau_{APLE, 1}$, and $\tau_{APLE, 0}$. The linearity assumption appears to be approximately satisfied for the treated but usually not for the untreated units.

Thus, as a robustness check, I also report a number of alternative estimates of the effects of NSW program in Table \ref{tab:nsw3}\@. I consider regression adjustment, as in section \ref{sec:theory}\ref{sec:example}, as well as matching on $p \left( X \right)$ and on the logit propensity score.\footnote{In particular, the estimates discussed in section \ref{sec:theory}\ref{sec:example} are reported in column 4 of the bottom panel of Table \ref{tab:nsw3}\@.} In each case, I separately estimate ATE, ATT, and ATU\@. These estimates are consistent with the claim that the general pattern of results in Table \ref{tab:nsw1} is driven by the OLS weights. The estimates of ATE and ATU are always negative and large in magnitude; the estimates of ATT are much closer to the experimental benchmark.

Finally, I repeat the following exercise from section \ref{sec:theory}\ref{sec:example}\@. When we match the OLS estimates in Table \ref{tab:nsw1} with the corresponding estimates of ATT and ATU in Table \ref{tab:nsw3}, we can write $\hat{\tau} = \hat{w}_{ATT} \cdot \hat{\tau}_{ATT} + \left( 1-\hat{w}_{ATT} \right) \cdot \hat{\tau}_{ATU}$. Unless $\hat{\tau}_{ATT}$ and $\hat{\tau}_{ATU}$ are sample analogues of $\tau_{APLE, 1}$ and $\tau_{APLE, 0}$, $\hat{w}_{ATT}$ does not need to be bounded between zero and one. Yet, we can solve for $\hat{w}_{ATT}$ for each set of estimates. The mean of $\hat{w}_{ATT}$ across all sets of estimates in Table \ref{tab:nsw3} is 98.3\%, which is nearly identical to the sample proportion of untreated units, 98.9\%. This is reassuring for my claims.

\subsection{The Effects of Cash Transfers on Longevity}
\label{sec:aizer}

In my second application, I replicate a recent paper by \cite{AEFLM2016} and study the effects of cash transfers on longevity of the children of their beneficiaries, as measured by their log age at death. In particular, \cite{AEFLM2016} analyze the administrative records of applicants to the Mothers' Pension (MP) program, which supported poor mothers with dependent children in pre-WWII United States. In this study, the untreated group consists only of children of mothers who applied for a transfer, were initially deemed eligible, but were ultimately rejected. This strategy is used to ensure that treated and untreated individuals are broadly comparable, and hence an ignorability assumption might be plausible. Nevertheless, rejected mothers were slightly older and came from slightly smaller and richer families than accepted mothers. Thus, as before, there is no reason to believe that ATT and ATU are equal, although it is perhaps less clear a priori which is larger. Unlike in section \ref{sec:empirical}\ref{sec:nsw}, it seems plausible that the researcher might be interested either in the average effect of cash transfers, ATE, or in their average effect for accepted applicants, ATT\@.

\begin{table}[!tb]
\begin{adjustwidth}{-1in}{-1in}
\centering
\begin{threeparttable}
\caption{\normalsize{The Effects of Cash Transfers on Longevity\label{tab:mp1}}}
\begin{normalsize}
\begin{tabular}{l >{\centering\arraybackslash}m{3.5cm} >{\centering\arraybackslash}m{3.5cm} >{\centering\arraybackslash}m{3.5cm} >{\centering\arraybackslash}m{3.5cm}}
\hline\hline
    \multicolumn{1}{c}{} & (1) & (2) & (3) & (4) \\
\hline
    \multicolumn{1}{c}{} & \multicolumn{4}{c}{Original estimates}  \\
\cline{2-5}
    \multicolumn{1}{l}{OLS} & 0.0157*** & 0.0158*** & 0.0182*** & 0.0167*** \\
    \multicolumn{1}{c}{} & (0.0058) & (0.0059) & (0.0062) & (0.0061) \\
    \multicolumn{1}{l}{} & & & & \\
    \multicolumn{1}{c}{} & \multicolumn{4}{c}{Diagnostics}    \\
\cline{2-5}
    \multicolumn{1}{l}{$\hat{w}_0$} &     0.861 & 0.870 & 0.784 & 0.784 \\
    \multicolumn{1}{l}{$\hat{w}_0^* = \hat{\rho}$} &     0.875 & 0.875 & 0.875 & 0.875 \\
    \multicolumn{1}{l}{$\hat{\delta}$} &     0.736 & 0.745 & 0.659 & 0.659 \\
    \multicolumn{1}{l}{$\hat{\delta}^* = 2 \hat{\rho} - 1$} &     0.750 & 0.750 & 0.750 & 0.750 \\
    \multicolumn{1}{l}{} & & & & \\
    \multicolumn{1}{c}{} & \multicolumn{4}{c}{Decomposition}    \\
\cline{2-5}
    \multicolumn{1}{l}{$\widehat{\mathrm{ATT}}$} & 0.0129** & 0.0149** & 0.0097 & 0.0089 \\
    \multicolumn{1}{l}{} & (0.0064) & (0.0071) & (0.0078) & (0.0079) \\
    \multicolumn{1}{l}{$\hat{w}_1$} & 0.139 & 0.130 & 0.216 & 0.216 \\
    \multicolumn{1}{l}{} & & & & \\
    \multicolumn{1}{l}{$\widehat{\mathrm{ATU}}$} & 0.0162*** & 0.0160*** & 0.0206*** & 0.0188*** \\
    \multicolumn{1}{l}{} & (0.0057) & (0.0059) & (0.0063) & (0.0064) \\
    \multicolumn{1}{l}{$\hat{w}_0$} & 0.861 & 0.870 & 0.784 & 0.784 \\
    \multicolumn{1}{l}{} & & & & \\
    \multicolumn{1}{l}{$\widehat{\mathrm{ATE}}$} & 0.0133** & 0.0150** & 0.0110 & 0.0102 \\
    \multicolumn{1}{l}{} & (0.0063) & (0.0068) & (0.0073) & (0.0074) \\
    \multicolumn{1}{l}{} & & & & \\
    \multicolumn{1}{l}{State fixed effects} & \checkmark & & & \\
    \multicolumn{1}{l}{County fixed effects} & & & \checkmark & \checkmark \\
    \multicolumn{1}{l}{Cohort fixed effects} & \checkmark & \checkmark & \checkmark & \checkmark \\
    \multicolumn{1}{l}{State characteristics} & & \checkmark & \checkmark & \checkmark \\
    \multicolumn{1}{l}{County characteristics} & & \checkmark & & \\
    \multicolumn{1}{l}{Individual characteristics} & & \checkmark & \checkmark & \checkmark \\
    \multicolumn{1}{l}{} & & & & \\
    \multicolumn{1}{l}{$\hat{\rho} = \hat{\pr} \left( d=1 \right)$} & 0.875 & 0.875 & 0.875 & 0.875 \\
    \multicolumn{1}{l}{Observations} & 7,860 & 7,859 & 7,859 & 7,857 \\
\hline
\end{tabular}
\end{normalsize}
\begin{footnotesize}
\begin{tablenotes}[flushleft]
\item \textit{Notes:} The estimates in the top panel correspond to columns 1 to 4 in panel A of Table 4 in \citet[][p.~952]{AEFLM2016}. The dependent variable is log age at death, as reported in the MP records (columns 1 to 3) or on the death certificate (column 4). State, county, and individual characteristics are listed in Table \ref{tab:mp3} in online appendix \ref{app:aizer}\@. Formulas for $w_0$, $w_1$, and $\delta$ are given in Theorem \ref{the:ols} and Corollary \ref{cor:biasate}. Following these results, $\mathrm{OLS} = \hat{w}_1 \cdot \widehat{\mathrm{ATT}} + \hat{w}_0 \cdot \widehat{\mathrm{ATU}}$\@. Estimates of ATE, ATT, and ATU are sample analogues of $\tau_{APLE}$, $\tau_{APLE, 1}$, and $\tau_{APLE, 0}$, respectively. Also, $\widehat{\mathrm{ATE}} = \hat{\rho} \cdot \widehat{\mathrm{ATT}} + \left( 1 - \hat{\rho} \right) \cdot \widehat{\mathrm{ATU}}$\@. Huber--White standard errors (OLS) and bootstrap standard errors ($\widehat{\mathrm{ATE}}$, $\widehat{\mathrm{ATT}}$, and $\widehat{\mathrm{ATU}}$) are in parentheses.
\item *Statistically significant at the 10\% level; **at the 5\% level; ***at the 1\% level.
\end{tablenotes}
\end{footnotesize}
\end{threeparttable}
\end{adjustwidth}
\end{table}

The top and middle panels of Table \ref{tab:mp1} reproduce the baseline estimates from \cite{AEFLM2016} and report my diagnostics. While the OLS estimates are positive and statistically significant, my diagnostics indicate that these results should be approached with caution. Namely, treated units constitute the vast majority (or 87.5\%) of the sample. It follows that OLS is expected to place a disproportionately large weight on $\widehat{\mathrm{ATU}}$, in which case the OLS estimates might be very biased for both ATE and ATT (cf.~Corollaries \ref{cor:biasate} and \ref{cor:biasatt}). Indeed, my estimates of $\delta$ suggest that the difference between the OLS estimand and ATE is equal to 65.9--74.5\% of the difference between ATU and ATT\@. Also, the estimates of $w_0$ suggest that the difference between OLS and ATT corresponds to 78.4--87.0\% of this measure of heterogeneity. The estimates of $\delta^*$ and $w_0^*$ are similar. It turns out that in this application the OLS estimates might be substantially biased for both of our parameters of interest. This would be a pessimistic scenario for OLS\@.

The results in the bottom panel of Table \ref{tab:mp1} suggest that these biases are indeed substantial. In this panel, following Corollary \ref{cor:causal}, each OLS estimate from \cite{AEFLM2016} is represented as a weighted average of estimates of two effects, on accepted (ATT) and rejected (ATU) applicants. The estimates of ATU are consistently larger than those of ATT\@. Thus, OLS overestimates both ATE (since $\hat{\delta} > 0$) and ATT\@. While the implicit OLS estimates of these parameters remain statistically significant in columns 1 and 2, this is no longer the case in columns 3 and 4, following the inclusion of county fixed effects. Perhaps more importantly, these estimates of ATT are half smaller than the corresponding OLS estimates. Clearly, this difference is economically quite meaningful.

To assess the robustness of these findings, I present several extensions of my analysis in online appendix \ref{app:aizer}\@. The informal test of Assumption \ref{ass:lin}, as discussed in section \ref{sec:empirical}\ref{sec:nsw}, appears to suggest that the conditional mean of the outcome given $p \left( X \right)$ is approximately linear for both the treated and untreated units (see Figures \ref{fig:mpypx1} to \ref{fig:mpypx4})\@. I also report a number of alternative estimates of the effects of cash transfers in Table \ref{tab:mp3}\@. These additional results support my conclusion. Only one in twelve estimates of ATT is statistically different from zero, and four of the insignificant estimates are negative. While it is possible that cash transfers increase longevity, the OLS estimates reported in \cite{AEFLM2016} are almost certainly too large. Interestingly, this bias appears to be driven by the implicit OLS weights on ATT and ATU, which were the focus of this paper.\footnote{I also repeat two further exercises from section \ref{sec:empirical}\ref{sec:nsw}\@. First, after I reestimate the model in (\ref{ols}) using WLS, with weights of 1 for treated and $\frac{1}{k}$ for untreated units, I demonstrate in Figures \ref{fig:mpwls1} to \ref{fig:mpwls4} that these estimates become more positive as $k$ increases. As before, larger values of $k$ translate into larger weights on $\widehat{\mathrm{ATU}}$, which is now greater than $\widehat{\mathrm{ATT}}$\@. Second, when I use the estimates of ATT and ATU in Table \ref{tab:mp3} to recover the hypothetical OLS weights, I obtain 22.8\% as the mean of $\hat{w}_{ATT}$. This is reasonably similar to the proportion of untreated units, 12.5\%.}

\section{Conclusion}
\label{sec:conclusion}

This paper proposed a new interpretation of the OLS estimand for the effect of a binary treatment in the standard linear model with additive effects. According to the main result of this paper, the OLS estimand is a convex combination of two parameters, which under certain conditions are equivalent to the average treatment effects on the treated (ATT) and untreated (ATU)\@. Surprisingly, the weights on these parameters are inversely related to the proportion of observations in each group, which can lead to substantial biases when interpreting the OLS estimand as ATE or ATT\@.

One lesson from this result is that it might be preferable, as suggested by a body of work in econometrics, to use any of the standard estimators of average treatment effects under ignorability, such as regression adjustment, weighting, matching, and various combinations of these approaches. Empirical researchers with a preference for OLS might instead want to use the diagnostic tools that this paper also provided. These diagnostics, which are implemented in the \texttt{\small{hettreatreg}} package in R and Stata, are applicable whenever the researcher is: \textit{(i)} studying the effects of a binary treatment, \textit{(ii)} using OLS, and \textit{(iii)} unwilling to maintain that ATT is exactly equal to ATU\@. In an important special case, these diagnostics only require the knowledge of the proportion of treated units.

\pagebreak

\begin{appendices}

\renewcommand{\thesection}{\Alph{section}}
\renewcommand{\thesubsection}{\Alph{section}\arabic{subsection}}

\onehalfspacing

\section*{Online Appendix}
\section{Proof of Theorem \ref{the:ols}}
\label{app:proofs}

\setcounter{equation}{0}
\renewcommand{\theequation}{\ref{app:proofs}\arabic{equation}}
\setcounter{lemma}{0}
\renewcommand{\thelemma}{\ref{app:proofs}\arabic{lemma}}

First, consider equation (\ref{lp_y}) in the main text, $\lp \left( y \mid 1, d, X \right) = \alpha + \tau d + X \beta$. By the Frisch--Waugh theorem, $\tau = \tau_a$, where $\tau_a$ is defined by
\begin{equation}
\label{aux}
\lp \left[ y \mid 1, d, p \left( X \right) \right] = \alpha_a + \tau_a d + \gamma_a \cdot p \left( X \right).
\end{equation}
Second, note that (\ref{aux}) is a linear projection of $y$ on two variables: one binary, $d$, and one arbitrarily discrete or continuous, $p \left( X \right)$. Thus, we can use the following result from \cite{EGH2010}.

\bigskip
\begin{lemma}[\citealp{EGH2010}]
\label{lem:egh}
Let $\lp \left( y \mid 1, d, x \right) = \alpha_e + \tau_e d + \beta_e x$ denote the linear projection of $y$ on $d$ (a binary variable) and $x$ (a single, possibly continuous variable). Then, assuming all objects are well defined,
\begingroup
\jot=10pt
\begin{eqnarray}
\tau_e & = & \frac{\rho \cdot \var \left( x \mid d=1 \right)}{\rho \cdot \var \left( x \mid d=1 \right) + \left( 1 - \rho \right) \cdot \var \left( x \mid d=0 \right)} \cdot \theta_1
\nonumber\\
& + & \frac{\left( 1 - \rho \right) \cdot \var \left( x \mid d=0 \right)}{\rho \cdot \var \left( x \mid d=1 \right) + \left( 1 - \rho \right) \cdot \var \left( x \mid d=0 \right)} \cdot \theta_0,
\nonumber
\end{eqnarray}
\endgroup
where
\begin{displaymath}
\theta_1 = \frac{\cov \left( d, y \right)}{\var \left( d \right)} - \frac{\cov \left( d, x \right)}{\var \left( d \right)} \cdot \frac{\cov \left( x, y \mid d=1 \right)}{\var \left( x \mid d=1 \right)}
\end{displaymath}
and
\begin{displaymath}
\theta_0 = \frac{\cov \left( d, y \right)}{\var \left( d \right)} - \frac{\cov \left( d, x \right)}{\var \left( d \right)} \cdot \frac{\cov \left( x, y \mid d=0 \right)}{\var \left( x \mid d=0 \right)}.
\end{displaymath}
\end{lemma}

\bigskip\bigskip
\noindent
Combining the two pieces gives
\begingroup
\jot=10pt
\begin{eqnarray}
\tau & = & \frac{\rho \cdot \var \left[ p \left( X \right) \mid d=1 \right]}{\rho \cdot \var \left[ p \left( X \right) \mid d=1 \right] + \left( 1 - \rho \right) \cdot \var \left[ p \left( X \right) \mid d=0 \right]} \cdot \theta^*_1
\nonumber\\
& + & \frac{\left( 1 - \rho \right) \cdot \var \left[ p \left( X \right) \mid d=0 \right]}{\rho \cdot \var \left[ p \left( X \right) \mid d=1 \right] + \left( 1 - \rho \right) \cdot \var \left[ p \left( X \right) \mid d=0 \right]} \cdot \theta^*_0,
\end{eqnarray}
\endgroup
where
\begin{equation}
\theta^*_1 = \frac{\cov \left( d, y \right)}{\var \left( d \right)} - \frac{\cov \left[ d, p \left( X \right) \right]}{\var \left( d \right)} \cdot \frac{\cov \left[ p \left( X \right), y \mid d=1 \right]}{\var \left[ p \left( X \right) \mid d=1 \right]}
\end{equation}
and
\begin{equation}
\theta^*_0 = \frac{\cov \left( d, y \right)}{\var \left( d \right)} - \frac{\cov \left[ d, p \left( X \right) \right]}{\var \left( d \right)} \cdot \frac{\cov \left[ p \left( X \right), y \mid d=0 \right]}{\var \left[ p \left( X \right) \mid d=0 \right]}.
\end{equation}
Third, notice that $\theta^*_1 = \tau_{APLE, 0}$ and $\theta^*_0 = \tau_{APLE, 1}$, as defined in equation (\ref{tau_apej}) in the main text. Indeed,
\begin{equation}
\frac{\cov \left( d, y \right)}{\var \left( d \right)} = \e \left( y \mid d=1 \right) - \e \left( y \mid d=0 \right)
\end{equation}
and also
\begin{equation}
\frac{\cov \left[ d, p \left( X \right) \right]}{\var \left( d \right)} = \e \left[ p \left( X \right) \mid d=1 \right] - \e \left[ p \left( X \right) \mid d=0 \right].
\end{equation}
Moreover, for $j=0,1$,
\begin{equation}
\frac{\cov \left[ p \left( X \right), y \mid d=j \right]}{\var \left[ p \left( X \right) \mid d=j \right]} = \gamma_j,
\end{equation}
where $\gamma_1$ and $\gamma_0$ are defined in equations (\ref{lp_y1}) and (\ref{lp_y0}) in the main text, respectively. Because
\begin{eqnarray}
\e \left( y \mid d=1 \right) - \e \left( y \mid d=0 \right) & = & \left\lbrace \e \left[ p \left( X \right) \mid d=1 \right] - \e \left[ p \left( X \right) \mid d=0 \right] \right\rbrace \cdot \gamma_1
\nonumber\\
& + & \left( \alpha_1 - \alpha_0 \right) + \left( \gamma_1 - \gamma_0 \right) \cdot \e \left[ p \left( X \right) \mid d=0 \right]
\label{ob1}
\end{eqnarray}
and also
\begin{eqnarray}
\e \left( y \mid d=1 \right) - \e \left( y \mid d=0 \right) & = & \left\lbrace \e \left[ p \left( X \right) \mid d=1 \right] - \e \left[ p \left( X \right) \mid d=0 \right] \right\rbrace \cdot \gamma_0
\nonumber\\
& + & \left( \alpha_1 - \alpha_0 \right) + \left( \gamma_1 - \gamma_0 \right) \cdot \e \left[ p \left( X \right) \mid d=1 \right],
\label{ob2}
\end{eqnarray}
where again $\alpha_1$ and $\alpha_0$ are defined in equations (\ref{lp_y1}) and (\ref{lp_y0}) in the main text, we get the result that $\theta^*_1 = \tau_{APLE, 0}$ and $\theta^*_0 = \tau_{APLE, 1}$. Note that equations (\ref{ob1}) and (\ref{ob2}) correspond to special cases of the Oaxaca--Blinder decomposition \citep{Blinder1973, Oaxaca1973, FLF2011}, which is also the focus of \cite{EGH2010}. Finally, combining the three pieces gives
\begingroup
\jot=10pt
\begin{eqnarray}
\tau & = & \frac{\rho \cdot \var \left[ p \left( X \right) \mid d=1 \right]}{\rho \cdot \var \left[ p \left( X \right) \mid d=1 \right] + \left( 1 - \rho \right) \cdot \var \left[ p \left( X \right) \mid d=0 \right]} \cdot \tau_{APLE, 0}
\nonumber\\
& + & \frac{\left( 1 - \rho \right) \cdot \var \left[ p \left( X \right) \mid d=0 \right]}{\rho \cdot \var \left[ p \left( X \right) \mid d=1 \right] + \left( 1 - \rho \right) \cdot \var \left[ p \left( X \right) \mid d=0 \right]} \cdot \tau_{APLE, 1},
\end{eqnarray}
\endgroup
which completes the proof.

\pagebreak

\section{Extensions}

\subsection{Proportion of Treated Units and OLS Weights}
\label{app:monotonic}

\setcounter{equation}{0}
\renewcommand{\theequation}{\ref{app:monotonic}.\arabic{equation}}
\setcounter{assumption}{0}
\renewcommand{\theassumption}{\ref{app:monotonic}.\arabic{assumption}}
\setcounter{proposition}{0}
\renewcommand{\theproposition}{\ref{app:monotonic}.\arabic{proposition}}

To show formally that $w_1$ is decreasing in $\rho$ and that $w_0$ is increasing in $\rho$, it is convenient to additionally assume that $\e \left( d \mid X \right)$ is linear in $X$.

\begin{assumption}
\label{ass:linpx}
$\e \left( d \mid X \right) = p \left( X \right) = \alpha_p + X \beta_p$.
\end{assumption}

\noindent
This restriction is satisfied automatically in saturated models, as studied by \cite{Angrist1998} and \cite{Humphreys2009}. It is also used by \cite{AS2016} and \cite{AAIW2020}. In the present context there are two reasons why Assumption \ref{ass:linpx} is useful. First, it allows us to rewrite $w_0$ and $w_1$ solely in terms of unconditional expectations of $p \left( X \right)$ and its powers. Second, it simplifies the calculation of the derivatives of $w_0$ and $w_1$ with respect to the intercept of the propensity score model. Imposing a shift on this intercept is equivalent to changing $\rho$ by a small amount. It turns out that Theorem \ref{the:ols} and Assumption \ref{ass:linpx} imply the following result.

\begin{proposition}
\label{prop:monotonic}
Under Assumptions \ref{ass:ols}, \ref{ass:px}, and \ref{ass:linpx},
\begin{eqnarray}
\frac{d w_1}{d \alpha_p} < 0 \text{\phantom{text} and \phantom{text}} \frac{d w_0}{d \alpha_p} > 0.
\nonumber
\end{eqnarray}
\end{proposition}
\begin{proof}
For simplicity, we first focus on $a_0$ and $a_1$, which we define as $a_0 = \rho \cdot \var \left[ p \left( X \right) \mid d=1 \right]$ and $a_1 = \left( 1 - \rho \right) \cdot \var \left[ p \left( X \right) \mid d=0 \right]$. As a result, $w_0 = \frac{a_0}{a_0 + a_1}$ and $w_1 = \frac{a_1}{a_0 + a_1}$. It turns out that we can rewrite $a_0$ as
\begin{eqnarray}
a_0 & = & \e \left( d \right) \cdot \e \left( \left\lbrace p \left( X \right) - \e \left[ p \left( X \right) \mid d=1 \right] \right\rbrace ^{2} \mid d=1 \right) \nonumber\\
& = & \e \left( d \right) \cdot \left( \e \left[ p \left( X \right) ^{2} \mid d=1 \right] - \left\lbrace \e \left[ p \left( X \right) \mid d=1 \right] \right\rbrace ^{2} \right) \nonumber\\
& = & \e \left( d \right) \cdot \left( \frac{\e \left[ p \left( X \right) ^{2} d \right]}{\e \left( d \right)} - \left\lbrace \frac{\e \left[ p \left( X \right) d \right]}{\e \left( d \right)} \right\rbrace ^{2} \right) \nonumber\\
& = & \e \left[ p \left( X \right) ^{2} d \right] - \frac{\left\lbrace \e \left[ p \left( X \right) d \right] \right\rbrace ^{2}}{\e \left( d \right)} \nonumber\\
& = & \e \left[ p \left( X \right) ^{2} \e \left( d \mid X \right) \right] - \frac{\left\lbrace \e \left[ p \left( X \right) \e \left( d \mid X \right) \right] \right\rbrace ^{2}}{\e \left[ \e \left( d \mid X \right) \right]} \nonumber\\
& = & \e \left[ p \left( X \right) ^{3} \right] - \frac{\left\lbrace \e \left[ p \left( X \right) ^{2} \right] \right\rbrace ^{2}}{\e \left[ p \left( X \right) \right]}.
\end{eqnarray}
Then, taking the derivative of $a_0$ with respect to $\alpha_p$ gives
\begin{eqnarray}
\frac{d a_0}{d \alpha_p} & = & 3 \e \left[ p \left( X \right) ^{2} \right] - \frac{4 \e \left[ p \left( X \right) ^{2} \right] \e \left[ p \left( X \right) \right]}{\e \left[ p \left( X \right) \right]} + \frac{\left\lbrace \e \left[ p \left( X \right) ^{2} \right] \right\rbrace ^{2}}{\e \left[ p \left( X \right) \right] ^{2}} \nonumber\\
& = & - \e \left[ p \left( X \right) ^{2} \right] + \frac{\left\lbrace \e \left[ p \left( X \right) ^{2} \right] \right\rbrace ^{2}}{\e \left[ p \left( X \right) \right] ^{2}} \nonumber\\
& = & \frac{\left\lbrace \e \left[ p \left( X \right) ^{2} \right] \right\rbrace ^{2} - \e \left[ p \left( X \right) ^{2} \right] \e \left[ p \left( X \right) \right] ^{2}}{\e \left[ p \left( X \right) \right] ^{2}} \nonumber\\
& = & \frac{\e \left[ p \left( X \right) ^{2} \right] \left\lbrace \e \left[ p \left( X \right) ^{2} \right] - \e \left[ p \left( X \right) \right] ^{2} \right\rbrace}{\e \left[ p \left( X \right) \right] ^{2}} \nonumber\\
& = & \frac{\e \left[ p \left( X \right) ^{2} \right] \var \left[ p \left( X \right) \right]}{\e \left[ p \left( X \right) \right] ^{2}} > 0.
\end{eqnarray}
Similarly,
\begin{eqnarray}
a_1 & = & \left[ 1 - \e \left( d \right) \right] \cdot \e \left( \left\lbrace p \left( X \right) - \e \left[ p \left( X \right) \mid d=0 \right] \right\rbrace ^{2} \mid d=0 \right) \nonumber\\
& = & \left[ 1 - \e \left( d \right) \right] \cdot \left( \e \left[ p \left( X \right) ^{2} \mid d=0 \right] - \left\lbrace \e \left[ p \left( X \right) \mid d=0 \right] \right\rbrace ^{2} \right) \nonumber\\
& = & \left[ 1 - \e \left( d \right) \right] \cdot \left( \frac{\e \left[ p \left( X \right) ^{2} \right] - \e \left[ p \left( X \right) ^{2} d \right]}{1 - \e \left( d \right)} - \left\lbrace \frac{\e \left[ p \left( X \right) \right] - \e \left[ p \left( X \right) d \right]}{1 - \e \left( d \right)} \right\rbrace ^{2} \right) \nonumber\\
& = & \e \left[ p \left( X \right) ^{2} \right] - \e \left[ p \left( X \right) ^{2} d \right] - \frac{\left\lbrace \e \left[ p \left( X \right) \right] - \e \left[ p \left( X \right) d \right] \right\rbrace ^{2}}{1 - \e \left( d \right)} \nonumber\\
& = & \e \left[ p \left( X \right) ^{2} \right] - \e \left[ p \left( X \right) ^{2} \e \left( d \mid X \right) \right] - \frac{\left\lbrace \e \left[ p \left( X \right) \right] - \e \left[ p \left( X \right) \e \left( d \mid X \right) \right] \right\rbrace ^{2}}{1 - \e \left[ \e \left( d \mid X \right) \right]} \nonumber\\
& = & \e \left[ p \left( X \right) ^{2} \right] - \e \left[ p \left( X \right) ^{3} \right] - \frac{\left\lbrace \e \left[ p \left( X \right) \right] - \e \left[ p \left( X \right) ^{2} \right] \right\rbrace ^{2}}{1 - \e \left[ p \left( X \right) \right]}
\end{eqnarray}
and
\begin{eqnarray}
\frac{d a_1}{d \alpha_p} & = & 2 \e \left[ p \left( X \right) \right] - 3 \e \left[ p \left( X \right) ^{2} \right] - \frac{\left\lbrace \e \left[ p \left( X \right) \right] - \e \left[ p \left( X \right) ^{2} \right] \right\rbrace ^{2}}{\left\lbrace 1 - \e \left[ p \left( X \right) \right] \right\rbrace ^{2}} \nonumber\\
& - & \frac{2 \cdot \left\lbrace 1 - \e \left[ p \left( X \right) \right] \right\rbrace \cdot \left\lbrace 1 - 2 \e \left[ p \left( X \right) \right] \right\rbrace \cdot \left\lbrace \e \left[ p \left( X \right) \right] - \e \left[ p \left( X \right) ^{2} \right] \right\rbrace}{\left\lbrace 1 - \e \left[ p \left( X \right) \right] \right\rbrace ^{2}} \nonumber\\
& = & \frac{\e \left[ p \left( X \right) \right] ^{2} - \e \left[ p \left( X \right) ^{2} \right]}{\left\lbrace 1 - \e \left[ p \left( X \right) \right] \right\rbrace ^{2}} \nonumber\\
& + & \frac{2 \e \left[ p \left( X \right) \right] \e \left[ p \left( X \right) ^{2} \right] - 2 \e \left[ p \left( X \right) \right] ^{3}}{\left\lbrace 1 - \e \left[ p \left( X \right) \right] \right\rbrace ^{2}} \nonumber\\
& + & \frac{\e \left[ p \left( X \right) ^{2} \right] \e \left[ p \left( X \right) \right] ^{2} - \left\lbrace \e \left[ p \left( X \right) ^{2} \right] \right\rbrace ^{2}}{\left\lbrace 1 - \e \left[ p \left( X \right) \right] \right\rbrace ^{2}} \nonumber\\
& = & \frac{- \var \left[ p \left( X \right) \right] \cdot \left\lbrace 1 - 2 \e \left[ p \left( X \right) \right] + \e \left[ p \left( X \right) ^{2} \right] \right\rbrace}{\left\lbrace 1 - \e \left[ p \left( X \right) \right] \right\rbrace ^{2}} \nonumber\\
& = & \frac{- \var \left[ p \left( X \right) \right] \cdot \e \left\lbrace \left[ 1 - p \left( X \right) \right] ^{2} \right\rbrace}{\left\lbrace 1 - \e \left[ p \left( X \right) \right] \right\rbrace ^{2}} < 0.
\end{eqnarray}
Finally, it follows that
\begin{eqnarray}
\frac{d w_1}{d \alpha_p} < 0 & \mathrm{and} & \frac{d w_0}{d \alpha_p} > 0,
\end{eqnarray}
since $w_0 = \frac{a_0}{a_0 + a_1}$, $w_1 = \frac{a_1}{a_0 + a_1}$, $a_0 > 0$, $a_1 > 0$, $\frac{d a_0}{d \alpha_p} > 0$, and $\frac{d a_1}{d \alpha_p} < 0$.
\end{proof}

\pagebreak

\subsection{Further Intuition for Theorem \ref{the:ols}}
\label{app:resid}

\setcounter{equation}{0}
\renewcommand{\theequation}{\ref{app:resid}.\arabic{equation}}
\setcounter{lemma}{0}
\renewcommand{\thelemma}{\ref{app:resid}.\arabic{lemma}}

We begin by noting that because the linear projection passes through the point of means of all variables, which implies, for example, that $\mathrm{E} \left( y \mid d=1 \right) = \alpha_1 + \gamma_1 \cdot \mathrm{E} \left[ p \left( X \right) \mid d=1 \right]$ and $\mathrm{E} \left( y \mid d=0 \right) = \alpha_0 + \gamma_0 \cdot \mathrm{E} \left[ p \left( X \right) \mid d=0 \right]$, the average partial linear effects of $d$ on both groups of interest can also be expressed as
\begin{equation}
\label{tau_ape1}
\tau_{APLE, 1} = \mathrm{E} \left( y \mid d=1 \right) - \left\lbrace \alpha_0 + \gamma_0 \cdot \mathrm{E} \left[ p \left( X \right) \mid d=1 \right] \right\rbrace
\end{equation}
and
\begin{equation}
\label{tau_ape0}
\tau_{APLE, 0} = \left\lbrace \alpha_1 + \gamma_1 \cdot \mathrm{E} \left[ p \left( X \right) \mid d=0 \right] \right\rbrace - \mathrm{E} \left( y \mid d=0 \right).
\end{equation}
In other words, we only need the linear projection of $y$ on $p \left( X \right)$ in group zero, and not in group one, to define $\tau_{APLE, 1}$. Similarly, we need the linear projection of $y$ on $p \left( X \right)$ in group one, but not in group zero, to define $\tau_{APLE, 0}$. When all objects are well defined, $\tau_{APLE, j}$ is also equivalent to the coefficient on $d$ in the linear projection of $y$ on $d$, $p \left( X \right)$, and $d \cdot \left\lbrace p \left( X \right) - \mathrm{E} \left[ p \left( X \right) \mid d=j \right] \right\rbrace$.

Then, an alternative intuition for the OLS weights in Theorem \ref{the:ols} follows from partial residualization that is implicit in least squares estimation. The first thing to note is that $\tau$, the OLS estimand, is equal to the coefficient on $d$ in the linear projection of $y - \gamma_a \cdot p \left( X \right)$ on $d$, where $\gamma_a$ is defined in equation (\ref{aux}). An implication of \cite{Deaton1997} and \cite{SHW2015} is that $\gamma_a$ is also a convex combination of $\gamma_1$ and $\gamma_0$, where the weight on $\gamma_1$ is \textit{increasing} in $\rho$. It follows that $\tau$ is a weighted average as well; it combines the coefficients on $d$ in the linear projections of $y - \gamma_1 \cdot p \left( X \right)$ and $y - \gamma_0 \cdot p \left( X \right)$ on $d$ in group zero and one, respectively. While the weight on the former (latter) is increasing (decreasing) in $\rho$, this parameter corresponds to $\tau_{APLE, 0}$ ($\tau_{APLE, 1}$), as can be seen from equations (\ref{tau_ape1}) and (\ref{tau_ape0}). Indeed, as noted above, it is $\gamma_1$ (and not $\gamma_0$) that is necessary to define $\tau_{APLE, 0}$. The bottom line is that when there are more treated than untreated units, $\gamma_1$ is likely to be better estimated than $\gamma_0$ and OLS gives more weight to the contrast of $y - \gamma_1 \cdot p \left( X \right)$, which in turn corresponds to $\tau_{APLE, 0}$. Interestingly, this parallels the intuition in \cite{Angrist1998} and \cite{AP2009} that OLS gives more weight to treatment effects that are better estimated in finite samples. Also, this discussion leads to an alternative proof of Theorem \ref{the:ols}.

\begin{proof}
As in online appendix \ref{app:proofs}, consider equation (\ref{lp_y}) in the main text, $\lp \left( y \mid 1, d, X \right) = \alpha + \tau d + X \beta$, and note that $\tau = \tau_a$, where $\tau_a$ is defined by $\lp \left[ y \mid 1, d, p \left( X \right) \right] = \alpha_a + \tau_a d + \gamma_a \cdot p \left( X \right)$. We can write this linear projection in error form as
\begin{equation}
\label{lp_jann}
y = \alpha_a + \tau_a d + \gamma_a \cdot p \left( X \right) + \upsilon.
\end{equation}
As in the main text, we also consider separate linear projections for $d=1$ and $d=0$, namely
\begin{equation}
\label{lp_y1_bis}
\lp \left[ y \mid 1, p \left( X \right), d=1 \right] = \alpha_1 + \gamma_1 \cdot p \left( X \right)
\end{equation}
and
\begin{equation}
\label{lp_y0_bis}
\lp \left[ y \mid 1, p \left( X \right), d=0 \right] = \alpha_0 + \gamma_0 \cdot p \left( X \right).
\end{equation}
Henceforth, to simplify notation I will use $l_1(X)$ to denote $\alpha_1 + \gamma_1 \cdot p \left( X \right)$ and $l_0(X)$ to denote $\alpha_0 + \gamma_0 \cdot p \left( X \right)$. To understand the relationship between $\gamma_a$, $\gamma_1$, and $\gamma_0$, we can use the following result from \cite{Deaton1997} and \cite{SHW2015}.

\bigskip
\begin{lemma}[\citealp{Deaton1997,SHW2015}]
\label{lem:deaton}
Let $\lp \left( y \mid 1, d, x \right) = \alpha_e + \tau_e d + \beta_e x$ denote the linear projection of $y$ on $d$ (a binary variable) and $x$ (a single, possibly continuous variable). Then, assuming all objects are well defined,
\begingroup
\jot=10pt
\begin{eqnarray}
\beta_e & = & \frac{\rho \cdot \var \left( x \mid d=1 \right)}{\rho \cdot \var \left( x \mid d=1 \right) + \left( 1 - \rho \right) \cdot \var \left( x \mid d=0 \right)} \cdot \beta_{1,e}
\nonumber\\
& + & \frac{\left( 1 - \rho \right) \cdot \var \left( x \mid d=0 \right)}{\rho \cdot \var \left( x \mid d=1 \right) + \left( 1 - \rho \right) \cdot \var \left( x \mid d=0 \right)} \cdot \beta_{0,e},
\nonumber
\end{eqnarray}
\endgroup
where $\beta_{1,e}$ and $\beta_{0,e}$ are defined by
\begin{displaymath}
\lp \left( y \mid 1, x, d=1 \right) = \alpha_{1,e} + \beta_{1,e} x
\end{displaymath}
and
\begin{displaymath}
\lp \left( y \mid 1, x, d=0 \right) = \alpha_{0,e} + \beta_{0,e} x.
\end{displaymath}
\end{lemma}

\bigskip\bigskip
\noindent
An implication of Lemma \ref{lem:deaton} is that
\begin{equation}
\gamma_a = w_0 \cdot \gamma_1 + w_1 \cdot \gamma_0.
\end{equation}
Next, we can rewrite equation (\ref{lp_jann}) as
\begin{eqnarray}
y - w_0 \cdot \gamma_1 \cdot p \left( X \right) - w_1 \cdot \gamma_0 \cdot p \left( X \right) &=& \alpha_a + \tau_a d + \upsilon
\nonumber\\
 &=& \e \left( y \right) - \tau_a \cdot \e \left( d \right) - \gamma_a \cdot \e \left[ p \left( X \right) \right] + \tau_a d + \upsilon.
\end{eqnarray}
Moreover, it turns out that
\begin{equation}
\alpha_1 = \e \left( y \mid d=1 \right) - \gamma_1 \cdot \e \left[ p \left( X \right) \mid d=1 \right]
\end{equation}
and also
\begin{equation}
\alpha_0 = \e \left( y \mid d=0 \right) - \gamma_0 \cdot \e \left[ p \left( X \right) \mid d=0 \right].
\end{equation}
It follows that
\begin{eqnarray}
y - w_0 \cdot l_1(X) - w_1 \cdot l_0(X) &=& \e \left( y \right) - w_0 \cdot \e \left( y \mid d=1 \right) - w_1 \cdot \e \left( y \mid d=0 \right)
\nonumber\\
 &+& w_0 \cdot \gamma_1 \cdot \left\lbrace \e \left[ p \left( X \right) \mid d=1 \right] - \e \left[ p \left( X \right) \right] \right\rbrace
\nonumber\\
 &+& w_1 \cdot \gamma_0 \cdot \left\lbrace \e \left[ p \left( X \right) \mid d=0 \right] - \e \left[ p \left( X \right) \right] \right\rbrace
\nonumber\\
 &-& \tau_a \cdot \e \left( d \right) + \tau_a d + \upsilon.
\end{eqnarray}
In other words, in a linear projection of $y - w_0 \cdot l_1(X) - w_1 \cdot l_0(X)$ on $d$, the coefficient on $d$ is equal to $\tau_a$ and the intercept is equal to $\e \left( y \right) - w_0 \cdot \e \left( y \mid d=1 \right) - w_1 \cdot \e \left( y \mid d=0 \right) + w_0 \cdot \gamma_1 \cdot \left\lbrace \e \left[ p \left( X \right) \mid d=1 \right] - \e \left[ p \left( X \right) \right] \right\rbrace + w_1 \cdot \gamma_0 \cdot \left\lbrace \e \left[ p \left( X \right) \mid d=0 \right] - \e \left[ p \left( X \right) \right] \right\rbrace - \tau_a \cdot \e \left( d \right)$. However, $\tau_a$ must also be equal to the difference in expected values of the dependent variable for $d=1$ and $d=0$. Using equations (\ref{tau_ape1}) and (\ref{tau_ape0}), we can write these expected values as
\begin{equation}
\e \left[ y - w_0 \cdot l_1(X) - w_1 \cdot l_0(X) \mid d=1 \right] = w_1 \cdot \tau_{APLE, 1}
\end{equation}
and
\begin{equation}
\e \left[ y - w_0 \cdot l_1(X) - w_1 \cdot l_0(X) \mid d=0 \right] = - w_0 \cdot \tau_{APLE, 0}.
\end{equation}
Thus,
\begin{equation}
\tau = \tau_a = w_1 \cdot \tau_{APLE, 1} + w_0 \cdot \tau_{APLE, 0},
\end{equation}
which completes the proof.
\end{proof}

\pagebreak

\subsection{A Weighted Least Squares Correction}
\label{app:wls}

\setcounter{equation}{0}
\renewcommand{\theequation}{\ref{app:wls}.\arabic{equation}}
\setcounter{proposition}{0}
\renewcommand{\theproposition}{\ref{app:wls}.\arabic{proposition}}

Suppose we use weighted least squares (WLS) to estimate the model with $d$ and $p \left( X \right)$ as the only independent variables. In this case we would like to obtain a set of weights, $w$, such that $\tau_w$ in
\begin{equation}
\mathrm{L} \left( \sqrt{w} \cdot y \mid \sqrt{w}, \sqrt{w} \cdot d, \sqrt{w} \cdot p \left( X \right) \right) = \alpha_w \sqrt{w} + \tau_w \sqrt{w} \cdot d + \gamma_w \sqrt{w} \cdot p \left( X \right)
\end{equation}
has a useful interpretation. An appropriate set of weights is provided in Proposition \ref{prop:wls}.

\medskip
\begin{proposition}[Weighted Least Squares Correction]
\label{prop:wls}
Suppose that Assumptions \ref{ass:ols} and \ref{ass:px} are satisfied. Also, $w = \frac{1 - \rho}{w_0} \cdot d + \frac{\rho}{w_1} \cdot \left( 1-d \right)$. Then,
\begin{eqnarray}
\tau_w = \tau_{APLE}.
\nonumber
\end{eqnarray}
Suppose that Assumptions \ref{ass:ols}, \ref{ass:px}, \ref{ass:uncon}, and \ref{ass:lin} are satisfied. Also, $w = \frac{1 - \rho}{w_0} \cdot d + \frac{\rho}{w_1} \cdot \left( 1-d \right)$. Then,
\begin{eqnarray}
\tau_w = \tau_{ATE}.
\nonumber
\end{eqnarray}
\end{proposition}

\medskip\medskip
\noindent
The proof of Proposition \ref{prop:wls} follows directly from the proofs of Theorem \ref{the:ols} and Corollary \ref{cor:causal}, and is omitted. Proposition \ref{prop:wls} establishes that the average effect of $d$ can be recovered from a weighted least squares procedure, with weights of $\frac{1 - \rho}{w_0}$ for units with $d=1$ and weights of $\frac{\rho}{w_1}$ for units with $d=0$. These weights consist of two parts: either $\frac{1}{w_1}$ or $\frac{1}{w_0}$; and either $\rho$ or $1 - \rho$. The role of the first part is always to undo the OLS weights ($w_1$ and $w_0$ in Theorem \ref{the:ols}); the role of the second part is to impose the correct weights of $\rho$ on the average effect of $d$ on group one and $1 - \rho$ on the average effect of $d$ on group zero. Finally, it is useful to note that there is no similar procedure to recover the average effects of $d$ on group zero and one; both of these objects, however, are easily obtained from equation (\ref{tau_apej}) in the main text.

Interestingly, the structure of the weighted least squares procedure in Proposition \ref{prop:wls} resembles the ``tyranny of the minority'' estimator in \cite{Lin2013}. This method uses weights of $\frac{1 - \rho}{\rho}$ for units with $d=1$ and weights of $\frac{\rho}{1 - \rho}$ for units with $d=0$; it also controls for $X$ instead of $p \left( X \right)$. It is important to note, however, that this method is designed to solve a different problem than Proposition \ref{prop:wls}. In particular, \cite{Freedman2008_Advances,Freedman2008_Annals} demonstrates that regression adjustments to experimental data can lead to a loss in precision. On the other hand, \cite{Lin2013} shows that this is no longer possible if we additionally interact $d$ with $X$ \citep[see also][]{NW2019}. Then, \cite{Lin2013} derives the ``tyranny of the minority'' estimator as an alternative procedure, based on a single conditional mean, which does not suffer from this loss in precision. In the context of observational data, however, the weights in \cite{Lin2013} guarantee that $\tau_w = \tau_{APLE}$ only in a special case, namely under Assumption \ref{ass:var}, $\var \left[p \left( X \right) \mid d=1 \right] = \var \left[ p \left( X \right) \mid d=0 \right]$.

\pagebreak

\subsection{Comparison with \cite{Angrist1998} and \cite{AS2016}}
\label{app:angrist}

\setcounter{equation}{0}
\renewcommand{\theequation}{\ref{app:angrist}.\arabic{equation}}

The result in \cite{Angrist1998} states that if $\lp \left( y \mid d, X \right) = \tau_n d + \sum_{s=1}^{S} \beta_{n,s} x_s$, where $X = \left( x_1, \ldots, x_S \right)$ is a vector of exhaustive and mutually exclusive ``stratum'' indicators, then
\begin{equation}
\tau_n = \sum_{s=1}^S \frac{\pr \left( x_s=1 \right) \cdot \pr \left( d=1 \mid x_s=1 \right) \cdot \pr \left( d=0 \mid x_s=1 \right)}{\sum_{t=1}^S \pr \left( x_t=1 \right) \cdot \pr \left( d=1 \mid x_t=1 \right) \cdot \pr \left( d=0 \mid x_t=1 \right)} \cdot \tau_s,
\label{angrist}
\end{equation}
where $\tau_s = \e \left( y \mid d=1, x_s=1 \right) - \e \left( y \mid d=0, x_s=1 \right)$. Further, under standard assumptions, $\tau_s = \e \left[ y(1)-y(0) \mid X \right]$. In this appendix I show that equation (\ref{angrist}) follows from Corollary \ref{cor:causal} when the model for $y$ is saturated.

The starting point is to note that, because the model for $y$ is saturated, $\e \left( d \mid X \right) = p \left( X \right) = \sum_{s=1}^{S} \beta_{p,s} x_s$. Additionally, Assumptions \ref{ass:uncon} and \ref{ass:lin} allow us to write $\e \left[ y(1)-y(0) \mid X \right] = \left( \alpha_1 - \alpha_0 \right) + \left( \gamma_1 - \gamma_0 \right) \cdot p \left( X \right)$. It follows that equation (\ref{angrist}) can alternatively be expressed as
\begingroup
\jot=10pt
\begin{eqnarray}
\tau_n &=& \frac{\e \left\lbrace p \left( X \right) \cdot \left[ 1-p \left( X \right) \right] \cdot \left[ \left( \alpha_1 - \alpha_0 \right) + \left( \gamma_1 - \gamma_0 \right) \cdot p \left( X \right) \right] \right\rbrace}{\e \left\lbrace p \left( X \right) \cdot \left[ 1-p \left( X \right) \right] \right\rbrace}
\nonumber\\
 &=& \left( \alpha_1 - \alpha_0 \right) + \left( \gamma_1 - \gamma_0 \right) \cdot \frac{\e \left[ p \left( X \right) ^{2} \right] - \e \left[ p \left( X \right) ^{3} \right]}{\e \left[ p \left( X \right) \right] - \e \left[ p \left( X \right) ^{2} \right]}.
\label{angrist_new}
\end{eqnarray}
\endgroup
The same representation of the OLS estimand under Assumptions \ref{ass:uncon} and \ref{ass:lin} follows from \cite{AS2016}, who generalize the result in \cite{Angrist1998} to any model, saturated or not, where $\e \left( d \mid X \right)$ is linear in $X$\@.

To demonstrate that the results in \cite{Angrist1998} and \cite{AS2016} follow from Corollary \ref{cor:causal}, we need to show that equation (\ref{angrist_new}) can be obtained by rearranging the expression in Corollary \ref{cor:causal}. To see this note that, under Assumptions \ref{ass:uncon} and \ref{ass:lin}, $\tau_{ATT} = \left( \alpha_1 - \alpha_0 \right) + \left( \gamma_1 - \gamma_0 \right) \cdot \e \left[ p \left( X \right) \mid d=1 \right]$ and $\tau_{ATU} = \left( \alpha_1 - \alpha_0 \right) + \left( \gamma_1 - \gamma_0 \right) \cdot \e \left[ p \left( X \right) \mid d=0 \right]$. Upon rearrangement,
\begingroup
\jot=10pt
\begin{eqnarray}
\tau_{ATT} &=& \left( \alpha_1 - \alpha_0 \right) + \left( \gamma_1 - \gamma_0 \right) \cdot \frac{\e \left[ p \left( X \right) d \right]}{\e \left( d \right)}
\nonumber\\
 &=& \left( \alpha_1 - \alpha_0 \right) + \left( \gamma_1 - \gamma_0 \right) \cdot \frac{\e \left[ p \left( X \right) ^{2} \right]}{\e \left[ p \left( X \right) \right]}
\end{eqnarray}
\endgroup
and
\begingroup
\jot=10pt
\begin{eqnarray}
\tau_{ATU} &=& \left( \alpha_1 - \alpha_0 \right) + \left( \gamma_1 - \gamma_0 \right) \cdot \frac{\e \left[ p \left( X \right) \cdot \left( 1-d \right) \right]}{1 - \e \left( d \right)}
\nonumber\\
 &=& \left( \alpha_1 - \alpha_0 \right) + \left( \gamma_1 - \gamma_0 \right) \cdot \frac{\e \left[ p \left( X \right) \right] - \e \left[ p \left( X \right) ^{2} \right]}{1 - \e \left[ p \left( X \right) \right]}.
\end{eqnarray}
\endgroup
Also, because $\e \left( d \mid X \right)$ is linear in $X$ and hence equal to $p \left( X \right)$, we can use the results from online appendix \ref{app:monotonic}, which state that $\rho \cdot \var \left[ p \left( X \right) \mid d=1 \right] = \e \left[ p \left( X \right) ^{3} \right] - \frac{\left\lbrace \e \left[ p \left( X \right) ^{2} \right] \right\rbrace ^{2}}{\e \left[ p \left( X \right) \right]}$ and $\left( 1 - \rho \right) \cdot \var \left[ p \left( X \right) \mid d=0 \right] = \e \left[ p \left( X \right) ^{2} \right] - \e \left[ p \left( X \right) ^{3} \right] - \frac{\left\lbrace \e \left[ p \left( X \right) \right] - \e \left[ p \left( X \right) ^{2} \right] \right\rbrace ^{2}}{1 - \e \left[ p \left( X \right) \right]}$. It follows that
\begingroup
\jot=10pt
\begin{eqnarray}
w_0 &=& \frac{\rho \cdot \var \left[ p \left( X \right) \mid d=1 \right]}{\rho \cdot \var \left[ p \left( X \right) \mid d=1 \right] + \left( 1 - \rho \right) \cdot \var \left[ p \left( X \right) \mid d=0 \right]}
\nonumber\\
 &=& \frac{\e \left[ p \left( X \right) ^{3} \right] - \frac{\left\lbrace \e \left[ p \left( X \right) ^{2} \right] \right\rbrace ^{2}}{\e \left[ p \left( X \right) \right]}}{\e \left[ p \left( X \right) ^{2} \right] - \frac{\left\lbrace \e \left[ p \left( X \right) ^{2} \right] \right\rbrace ^{2}}{\e \left[ p \left( X \right) \right]} - \frac{\left\lbrace \e \left[ p \left( X \right) \right] - \e \left[ p \left( X \right) ^{2} \right] \right\rbrace ^{2}}{1 - \e \left[ p \left( X \right) \right]}}
\end{eqnarray}
\endgroup
and
\begingroup
\jot=10pt
\begin{eqnarray}
w_1 &=& \frac{\left( 1 - \rho \right) \cdot \var \left[ p \left( X \right) \mid d=0 \right]}{\rho \cdot \var \left[ p \left( X \right) \mid d=1 \right] + \left( 1 - \rho \right) \cdot \var \left[ p \left( X \right) \mid d=0 \right]}
\nonumber\\
 &=& \frac{\e \left[ p \left( X \right) ^{2} \right] - \e \left[ p \left( X \right) ^{3} \right] - \frac{\left\lbrace \e \left[ p \left( X \right) \right] - \e \left[ p \left( X \right) ^{2} \right] \right\rbrace ^{2}}{1 - \e \left[ p \left( X \right) \right]}}{\e \left[ p \left( X \right) ^{2} \right] - \frac{\left\lbrace \e \left[ p \left( X \right) ^{2} \right] \right\rbrace ^{2}}{\e \left[ p \left( X \right) \right]} - \frac{\left\lbrace \e \left[ p \left( X \right) \right] - \e \left[ p \left( X \right) ^{2} \right] \right\rbrace ^{2}}{1 - \e \left[ p \left( X \right) \right]}}.
\end{eqnarray}
\endgroup
Consequently, an implication of Corollary \ref{cor:causal} is that
\begingroup
\jot=10pt
\begin{eqnarray}
\tau_n &=& w_1 \cdot \tau_{ATT} + w_0 \cdot \tau_{ATU}
\nonumber\\
 &=& \left( \alpha_1 - \alpha_0 \right) + \left( \gamma_1 - \gamma_0 \right) \cdot \frac{\left\lbrace \e \left[ p \left( X \right) ^{2} \right] - \e \left[ p \left( X \right) ^{3} \right] - \frac{\left\lbrace \e \left[ p \left( X \right) \right] - \e \left[ p \left( X \right) ^{2} \right] \right\rbrace ^{2}}{1 - \e \left[ p \left( X \right) \right]} \right\rbrace \cdot \e \left[ p \left( X \right) ^{2} \right]}{\left\lbrace \e \left[ p \left( X \right) ^{2} \right] - \frac{\left\lbrace \e \left[ p \left( X \right) ^{2} \right] \right\rbrace ^{2}}{\e \left[ p \left( X \right) \right]} - \frac{\left\lbrace \e \left[ p \left( X \right) \right] - \e \left[ p \left( X \right) ^{2} \right] \right\rbrace ^{2}}{1 - \e \left[ p \left( X \right) \right]} \right\rbrace \cdot \e \left[ p \left( X \right) \right]}
\nonumber\\
 &+& \left( \gamma_1 - \gamma_0 \right) \cdot \frac{\left\lbrace \e \left[ p \left( X \right) ^{3} \right] - \frac{\left\lbrace \e \left[ p \left( X \right) ^{2} \right] \right\rbrace ^{2}}{\e \left[ p \left( X \right) \right]} \right\rbrace \cdot \left\lbrace \e \left[ p \left( X \right) \right] - \e \left[ p \left( X \right) ^{2} \right] \right\rbrace}{\left\lbrace \e \left[ p \left( X \right) ^{2} \right] - \frac{\left\lbrace \e \left[ p \left( X \right) ^{2} \right] \right\rbrace ^{2}}{\e \left[ p \left( X \right) \right]} - \frac{\left\lbrace \e \left[ p \left( X \right) \right] - \e \left[ p \left( X \right) ^{2} \right] \right\rbrace ^{2}}{1 - \e \left[ p \left( X \right) \right]} \right\rbrace \cdot \left\lbrace 1 - \e \left[ p \left( X \right) \right] \right\rbrace}
\end{eqnarray}
\endgroup
or, equivalently,
\begin{equation}
\tau_n = \left( \alpha_1 - \alpha_0 \right) + \left( \gamma_1 - \gamma_0 \right) \cdot \frac{\lambda_n}{\lambda_d},
\label{enter_lambdas}
\end{equation}
where
\begin{equation}
\lambda_d = \left\lbrace \e \left[ p \left( X \right) ^{2} \right] - \frac{\left\lbrace \e \left[ p \left( X \right) ^{2} \right] \right\rbrace ^{2}}{\e \left[ p \left( X \right) \right]} - \frac{\left\lbrace \e \left[ p \left( X \right) \right] - \e \left[ p \left( X \right) ^{2} \right] \right\rbrace ^{2}}{1 - \e \left[ p \left( X \right) \right]} \right\rbrace \cdot \e \left[ p \left( X \right) \right] \cdot \left\lbrace 1 - \e \left[ p \left( X \right) \right] \right\rbrace
\end{equation}
and
\begin{eqnarray}
\lambda_n &=& \left\lbrace \e \left[ p \left( X \right) ^{2} \right] - \e \left[ p \left( X \right) ^{3} \right] - \frac{\left\lbrace \e \left[ p \left( X \right) \right] - \e \left[ p \left( X \right) ^{2} \right] \right\rbrace ^{2}}{1 - \e \left[ p \left( X \right) \right]} \right\rbrace \cdot \e \left[ p \left( X \right) ^{2} \right] \cdot \left\lbrace 1 - \e \left[ p \left( X \right) \right] \right\rbrace
\nonumber\\
 &+& \left\lbrace \e \left[ p \left( X \right) ^{3} \right] - \frac{\left\lbrace \e \left[ p \left( X \right) ^{2} \right] \right\rbrace ^{2}}{\e \left[ p \left( X \right) \right]} \right\rbrace \cdot \left\lbrace \e \left[ p \left( X \right) \right] - \e \left[ p \left( X \right) ^{2} \right] \right\rbrace \cdot \e \left[ p \left( X \right) \right].
\end{eqnarray}
Upon further rearrangement,
\begingroup
\jot=10pt
\begin{eqnarray}
\lambda_d &=& \e \left[ p \left( X \right) ^{2} \right] \cdot \e \left[ p \left( X \right) \right] + \e \left[ p \left( X \right) ^{2} \right] \cdot \left\lbrace \e \left[ p \left( X \right) \right] \right\rbrace ^{2} - \left\lbrace \e \left[ p \left( X \right) ^{2} \right] \right\rbrace ^{2} - \left\lbrace \e \left[ p \left( X \right) \right] \right\rbrace ^{3}
\nonumber\\
 &=& \left\lbrace \e \left[ p \left( X \right) \right] - \e \left[ p \left( X \right) ^{2} \right] \right\rbrace \cdot \left\lbrace \e \left[ p \left( X \right) ^{2} \right] - \left\lbrace \e \left[ p \left( X \right) \right] \right\rbrace ^{2} \right\rbrace
\label{lambda_d}
\end{eqnarray}
\endgroup
and
\begingroup
\jot=10pt
\begin{eqnarray}
\lambda_n &=& \left\lbrace \e \left[ p \left( X \right) ^{2} \right] \right\rbrace ^{2} + \e \left[ p \left( X \right) ^{3} \right] \cdot \left\lbrace \e \left[ p \left( X \right) \right] \right\rbrace ^{2} - \e \left[ p \left( X \right) ^{3} \right] \cdot \e \left[ p \left( X \right) ^{2} \right] - \e \left[ p \left( X \right) ^{2} \right] \cdot \left\lbrace \e \left[ p \left( X \right) \right] \right\rbrace ^{2}
\nonumber\\
 &=& \left\lbrace \e \left[ p \left( X \right) ^{2} \right] - \e \left[ p \left( X \right) ^{3} \right] \right\rbrace \cdot \left\lbrace \e \left[ p \left( X \right) ^{2} \right] - \left\lbrace \e \left[ p \left( X \right) \right] \right\rbrace ^{2} \right\rbrace.
\label{lambda_n}
\end{eqnarray}
\endgroup
Finally, plugging equations (\ref{lambda_d}) and (\ref{lambda_n}) into equation (\ref{enter_lambdas}) gives
\begingroup
\jot=10pt
\begin{eqnarray}
\tau_n &=& \left( \alpha_1 - \alpha_0 \right) + \left( \gamma_1 - \gamma_0 \right) \cdot \frac{\left\lbrace \e \left[ p \left( X \right) ^{2} \right] - \e \left[ p \left( X \right) ^{3} \right] \right\rbrace \cdot \left\lbrace \e \left[ p \left( X \right) ^{2} \right] - \left\lbrace \e \left[ p \left( X \right) \right] \right\rbrace ^{2} \right\rbrace}{\left\lbrace \e \left[ p \left( X \right) \right] - \e \left[ p \left( X \right) ^{2} \right] \right\rbrace \cdot \left\lbrace \e \left[ p \left( X \right) ^{2} \right] - \left\lbrace \e \left[ p \left( X \right) \right] \right\rbrace ^{2} \right\rbrace}
\nonumber\\
 &=& \left( \alpha_1 - \alpha_0 \right) + \left( \gamma_1 - \gamma_0 \right) \cdot \frac{\e \left[ p \left( X \right) ^{2} \right] - \e \left[ p \left( X \right) ^{3} \right]}{\e \left[ p \left( X \right) \right] - \e \left[ p \left( X \right) ^{2} \right]}.
\label{tau_final}
\end{eqnarray}
\endgroup
The equivalence between equations (\ref{angrist_new}) and (\ref{tau_final}) confirms that the result in \cite{Angrist1998} follows from Corollary \ref{cor:causal} when the model for $y$ is saturated. Similarly, the result in \cite{AS2016} follows from Corollary \ref{cor:causal} when $\e \left( d \mid X \right)$ is linear in $X$\@.

\pagebreak

\section{Implementation in Stata}
\label{app:stata}

This appendix discusses the implementation of my theoretical results using the Stata package \texttt{\small{hettreatreg}}. In particular, I show how to apply this package to obtain the estimates in column 4 of Table \ref{tab:nsw1} in the main text. To download this package and the NSW--CPS data from SSC, type
\begin{small}
\begin{verbatim}
. ssc install hettreatreg, all
\end{verbatim}
\end{small}
in the Command window. Then, type
\begin{small}
\begin{verbatim}
. use nswcps, clear
\end{verbatim}
\end{small}
to open the NSW--CPS data set. Then, the standard way to obtain the OLS estimate in column 4 of Table \ref{tab:nsw1} in the main text would be to type
\begin{small}
\begin{verbatim}
. regress re78 treated age-re75, vce(robust)

Linear regression                               Number of obs     =     16,177
                                                F(10, 16166)      =    1718.20
                                                Prob > F          =     0.0000
                                                R-squared         =     0.4762
                                                Root MSE          =     7001.7

------------------------------------------------------------------------------
             |               Robust
        re78 |      Coef.   Std. Err.      t    P>|t|     [95% Conf. Interval]
-------------+----------------------------------------------------------------
     treated |    793.587   618.6092     1.28   0.200    -418.9555     2006.13
         age |  -233.6775    40.7162    -5.74   0.000    -313.4857   -153.8692
        age2 |   1.814371   .5581946     3.25   0.001     .7202474    2.908494
        educ |   166.8492   28.70683     5.81   0.000     110.5807    223.1178
       black |  -790.6086   197.8149    -4.00   0.000    -1178.348   -402.8694
    hispanic |  -175.9751   218.3033    -0.81   0.420    -603.8738    251.9235
     married |    224.266   152.4363     1.47   0.141    -74.52594    523.0579
    nodegree |   311.8445    176.414     1.77   0.077     -33.9464    657.6355
        re74 |   .2953363   .0152084    19.42   0.000     .2655261    .3251466
        re75 |   .4706353   .0153101    30.74   0.000     .4406259    .5006447
       _cons |   7634.344   737.8143    10.35   0.000     6188.146    9080.542
------------------------------------------------------------------------------
\end{verbatim}
\end{small}
\bigskip\bigskip
It is also possible, however, to obtain the same output and several additional estimates---including those of my diagnostics and those of implicit estimates of ATE, ATT, and ATU---by typing
\begin{small}
\begin{verbatim}
. hettreatreg age-re75, o(re78) t(treated) noisily vce(robust)

Linear regression                               Number of obs     =     16,177
                                                F(10, 16166)      =    1718.20
                                                Prob > F          =     0.0000
                                                R-squared         =     0.4762
                                                Root MSE          =     7001.7

------------------------------------------------------------------------------
             |               Robust
        re78 |      Coef.   Std. Err.      t    P>|t|     [95% Conf. Interval]
-------------+----------------------------------------------------------------
     treated |    793.587   618.6092     1.28   0.200    -418.9555     2006.13
         age |  -233.6775    40.7162    -5.74   0.000    -313.4857   -153.8692
        age2 |   1.814371   .5581946     3.25   0.001     .7202474    2.908494
        educ |   166.8492   28.70683     5.81   0.000     110.5807    223.1178
       black |  -790.6086   197.8149    -4.00   0.000    -1178.348   -402.8694
    hispanic |  -175.9751   218.3033    -0.81   0.420    -603.8738    251.9235
     married |    224.266   152.4363     1.47   0.141    -74.52594    523.0579
    nodegree |   311.8445    176.414     1.77   0.077     -33.9464    657.6355
        re74 |   .2953363   .0152084    19.42   0.000     .2655261    .3251466
        re75 |   .4706353   .0153101    30.74   0.000     .4406259    .5006447
       _cons |   7634.344   737.8143    10.35   0.000     6188.146    9080.542
------------------------------------------------------------------------------

"OLS" is the estimated regression coefficient on treated.

   OLS  =  793.6    

P(d=1)  =  .011
P(d=0)  =  .989

    w1  =  .983
    w0  =  .017
 delta  =  -.971

   ATE  =  -6751    
   ATT  =  928.4    
   ATU  =  -6840    

OLS = w1*ATT + w0*ATU = 793.6
\end{verbatim}
\end{small}
\bigskip\bigskip
Alternatively, we may restrict our attention to this additional output by typing
\begin{small}
\begin{verbatim}
. hettreatreg age-re75, o(re78) t(treated)

"OLS" is the estimated regression coefficient on treated.

   OLS  =  793.6    

P(d=1)  =  .011
P(d=0)  =  .989

    w1  =  .983
    w0  =  .017
 delta  =  -.971

   ATE  =  -6751    
   ATT  =  928.4    
   ATU  =  -6840    

OLS = w1*ATT + w0*ATU = 793.6
\end{verbatim}
\end{small}
\bigskip\bigskip
In any case, \texttt{\small{OLS}} is the estimated regression coefficient on the variable designated as treatment. \texttt{\small{P(d=1)}} and \texttt{\small{P(d=0)}} correspond to $\hat{\rho}$ and $1-\hat{\rho}$, respectively. \texttt{\small{w1}}, \texttt{\small{w0}}, and \texttt{\small{delta}} correspond to $\hat{w}_1$, $\hat{w}_0$, and $\hat{\delta}$, respectively. Finally, \texttt{\small{ATE}}, \texttt{\small{ATT}}, and \texttt{\small{ATU}} correspond to $\hat{\tau}_{APLE}$, $\hat{\tau}_{APLE, 1}$, and $\hat{\tau}_{APLE, 0}$, respectively. \texttt{\small{hettreatreg}} stores all these estimates in \texttt{\small{e()}}. Type
\begin{small}
\begin{verbatim}
. help hettreatreg
\end{verbatim}
\end{small}
for more information and additional examples.

\pagebreak

\section{Implementation in R}
\label{app:r}

Similar to online appendix \ref{app:stata}, I now discuss the implementation of my theoretical results using the R package \texttt{\small{hettreatreg}}. As before, I show how to obtain the estimates reported in column 4 of Table \ref{tab:nsw1} in the main text. To download this package and the NSW--CPS data from CRAN, type
\begin{small}
\begin{verbatim}
> install.packages("hettreatreg")
\end{verbatim}
\end{small}
in the R/R Studio console. Next, type
\begin{small}
\begin{verbatim}
> library(hettreatreg)
\end{verbatim}
\end{small}
to load \texttt{\small{hettreatreg}} and
\begin{small}
\begin{verbatim}
> data("nswcps")
\end{verbatim}
\end{small}
to open the NSW--CPS data set. Then, the standard way to obtain the OLS estimate in column 4 of Table \ref{tab:nsw1} in the main text would be to type
\begin{small}
\begin{verbatim}
> lm(re78 ~ treated + age + age2 + educ + black + hispanic + married + nodegree
 + re74 + re75, data = nswcps)

Call:
lm(formula = re78 ~ treated + age + age2 + educ + black + hispanic + 
    married + nodegree + re74 + re75, data = nswcps)

Coefficients:
(Intercept)      treated          age         age2         educ        black  
  7634.3441     793.5870    -233.6775       1.8144     166.8492    -790.6086  
   hispanic      married     nodegree         re74         re75  
  -175.9751     224.2660     311.8445       0.2953       0.4706  
\end{verbatim}
\end{small}
\bigskip\bigskip
Using \texttt{\small{hettreatreg}}, it is possible to obtain several additional estimates, including those of my diagnostics and those of implicit estimates of ATE, ATT, and ATU\@. Before doing so, it is useful to designate an outcome variable, a treatment variable, and a list of control variables. To do this, type
\begin{small}
\begin{verbatim}
> outcome <- nswcps$re78
> treated <- nswcps$treated
> our_vars <- c("age", "age2", "educ", "black", "hispanic", "married", "nodegree",
 "re74", "re75")
> covariates <- subset(nswcps, select = our_vars)
\end{verbatim}
\end{small}
Then, type
\begin{small}
\begin{verbatim}
> hettreatreg(outcome, treated, covariates, verbose = TRUE)

"OLS" is the estimated regression coefficient on treated. 
 
   OLS = 793.6 
 
P(d=1) = 0.011 
P(d=0) = 0.989 
 
    w1 = 0.983 
    w0 = 0.017 
 delta = -0.971 
 
   ATE = -6751 
   ATT = 928.4 
   ATU = -6840 
 
OLS = w1*ATT + w0*ATU = 793.6 
\end{verbatim}
\end{small}
\bigskip\bigskip
To interpret these estimates, note that \texttt{\small{OLS}} is the estimated regression coefficient on the variable designated as treatment. \texttt{\small{P(d=1)}} and \texttt{\small{P(d=0)}} correspond to $\hat{\rho}$ and $1-\hat{\rho}$, respectively. \texttt{\small{w1}}, \texttt{\small{w0}}, and \texttt{\small{delta}} correspond to $\hat{w}_1$, $\hat{w}_0$, and $\hat{\delta}$, respectively. Finally, \texttt{\small{ATE}}, \texttt{\small{ATT}}, and \texttt{\small{ATU}} correspond to $\hat{\tau}_{APLE}$, $\hat{\tau}_{APLE, 1}$, and $\hat{\tau}_{APLE, 0}$, respectively. Type
\begin{small}
\begin{verbatim}
> ?hettreatreg
\end{verbatim}
\end{small}
for more information and an additional example. Further information is also available from CRAN at https://CRAN.R-project.org/package=hettreatreg.

\pagebreak

\section{Robustness Checks}
\label{app:robustness}

\subsection{The Effects of a Training Program on Earnings}
\label{app:mhe}

\setcounter{table}{0}
\renewcommand{\thetable}{\ref{app:mhe}.\arabic{table}}
\setcounter{figure}{0}
\renewcommand{\thefigure}{\ref{app:mhe}.\arabic{figure}}

\vspace{3.2cm}

\begin{figure}[!h]
\begin{adjustwidth}{-1in}{-1in}
\centering
\caption{WLS Estimates of the Effects of a Training Program on Earnings\label{fig:nswwls1}}
\includegraphics[width=12cm]{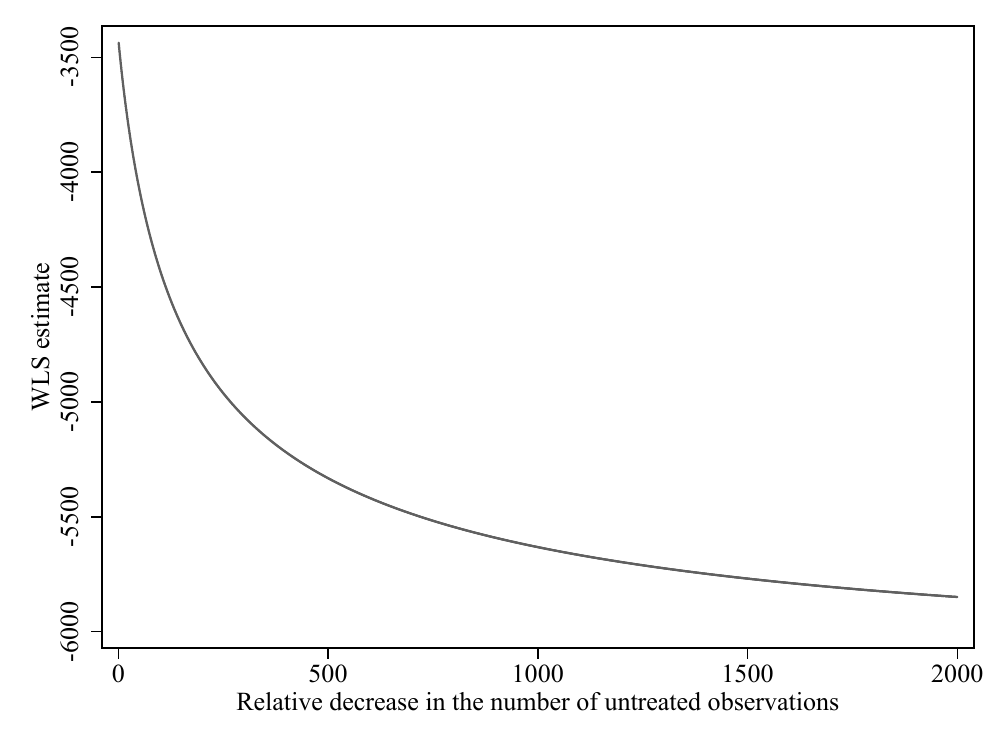}
\\
\vspace{-2mm}
\begin{footnotesize}
\begin{tabular}{p{12cm}}
\textit{Notes:} The vertical axis represents WLS estimates of the effect of NSW program on earnings in 1978 using the model in equation (\ref{ols}) in the main text and the specification in column 1 of Table \ref{tab:nsw1}, with weights of 1 for treated and $\frac{1}{k}$ for untreated units. The horizontal axis represents $k$.
\end{tabular}
\end{footnotesize}
\end{adjustwidth}
\end{figure}

\begin{figure}[!h]
\begin{adjustwidth}{-1in}{-1in}
\centering
\caption{WLS Estimates of the Effects of a Training Program on Earnings\label{fig:nswwls2}}
\includegraphics[width=12cm]{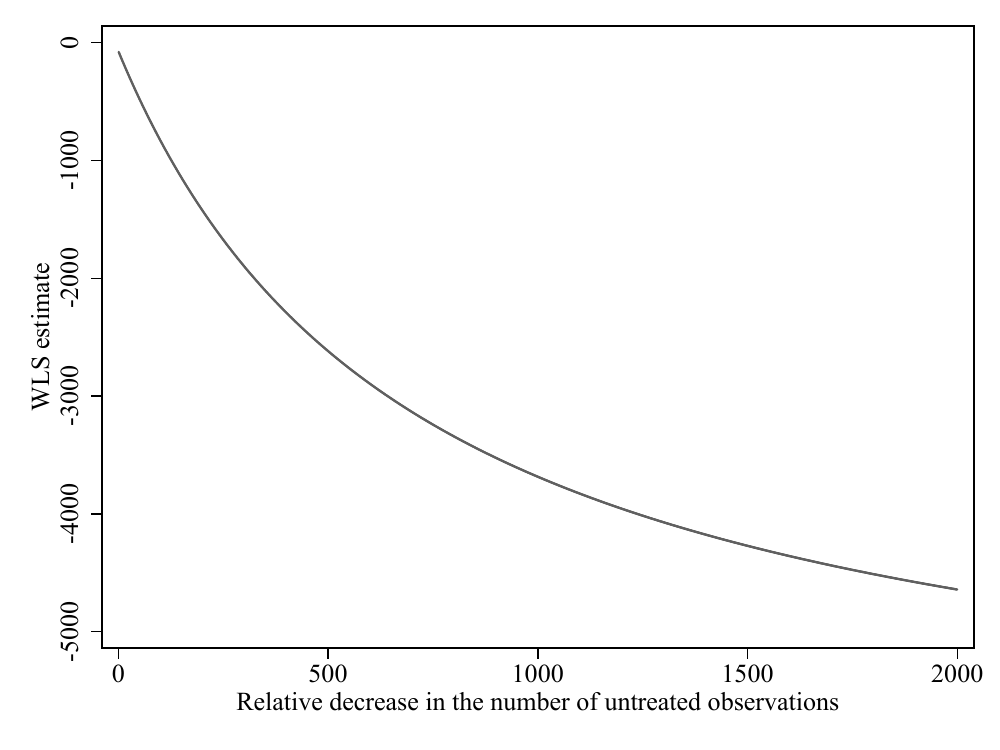}
\\
\vspace{-2mm}
\begin{footnotesize}
\begin{tabular}{p{12cm}}
\textit{Notes:} The vertical axis represents WLS estimates of the effect of NSW program on earnings in 1978 using the model in equation (\ref{ols}) in the main text and the specification in column 2 of Table \ref{tab:nsw1}, with weights of 1 for treated and $\frac{1}{k}$ for untreated units. The horizontal axis represents $k$.
\end{tabular}
\end{footnotesize}
\end{adjustwidth}
\end{figure}

\begin{figure}[!h]
\begin{adjustwidth}{-1in}{-1in}
\centering
\caption{WLS Estimates of the Effects of a Training Program on Earnings\label{fig:nswwls3}}
\includegraphics[width=12cm]{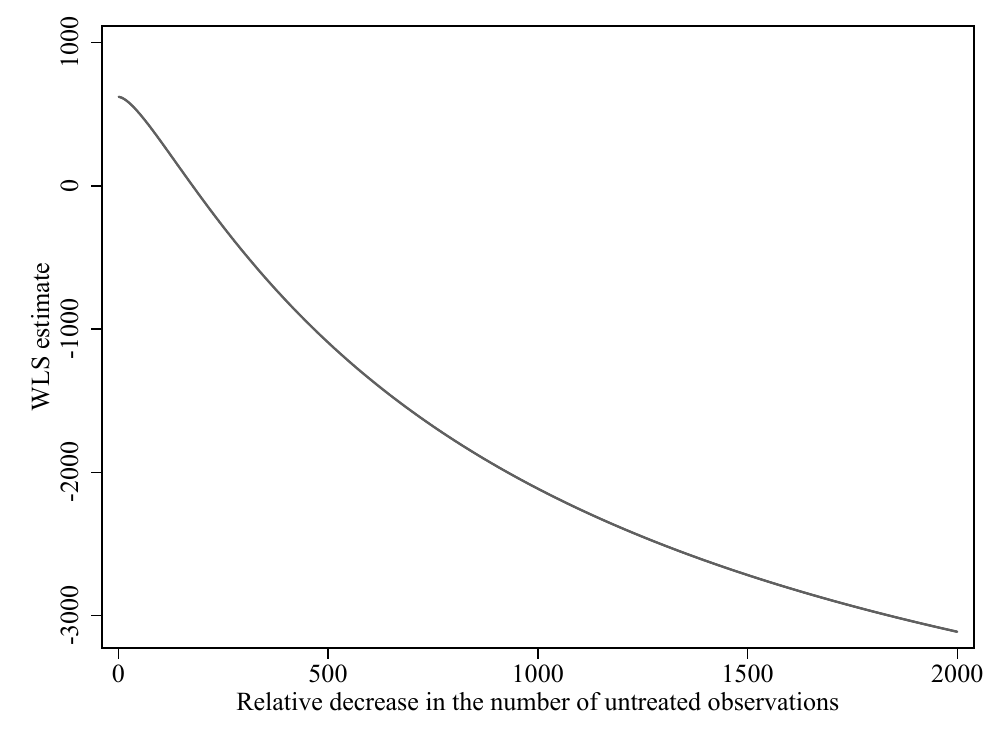}
\\
\vspace{-2mm}
\begin{footnotesize}
\begin{tabular}{p{12cm}}
\textit{Notes:} The vertical axis represents WLS estimates of the effect of NSW program on earnings in 1978 using the model in equation (\ref{ols}) in the main text and the specification in column 3 of Table \ref{tab:nsw1}, with weights of 1 for treated and $\frac{1}{k}$ for untreated units. The horizontal axis represents $k$.
\end{tabular}
\end{footnotesize}
\end{adjustwidth}
\end{figure}

\begin{figure}[!h]
\begin{adjustwidth}{-1in}{-1in}
\centering
\caption{WLS Estimates of the Effects of a Training Program on Earnings\label{fig:nswwls4}}
\includegraphics[width=12cm]{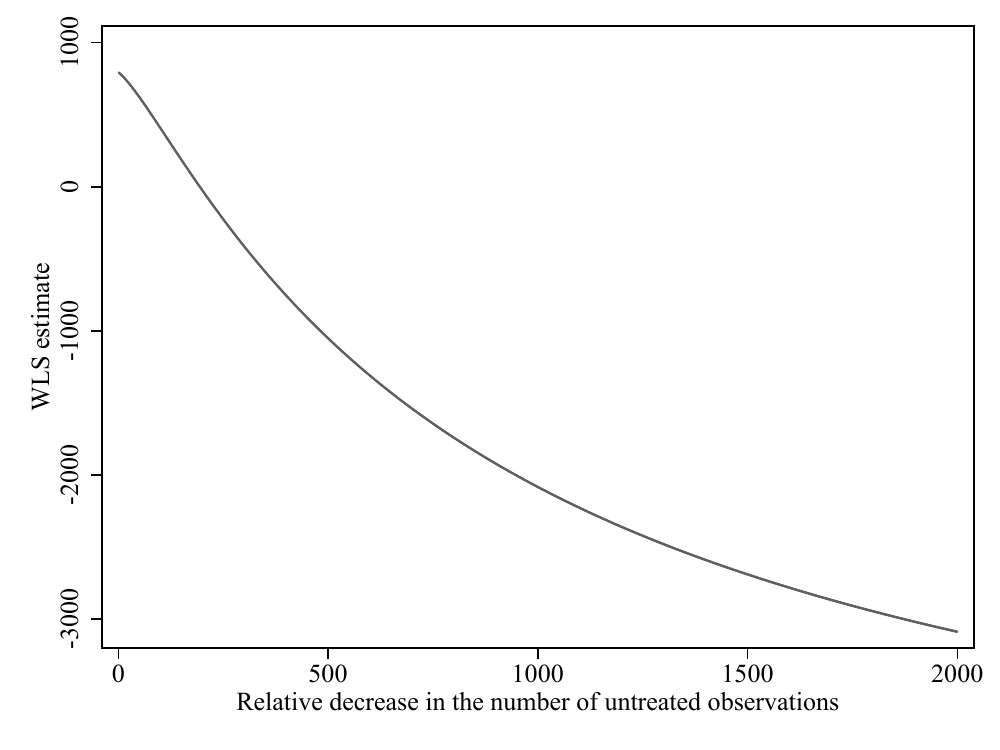}
\\
\vspace{-2mm}
\begin{footnotesize}
\begin{tabular}{p{12cm}}
\textit{Notes:} The vertical axis represents WLS estimates of the effect of NSW program on earnings in 1978 using the model in equation (\ref{ols}) in the main text and the specification in column 4 of Table \ref{tab:nsw1}, with weights of 1 for treated and $\frac{1}{k}$ for untreated units. The horizontal axis represents $k$.
\end{tabular}
\end{footnotesize}
\end{adjustwidth}
\end{figure}

\begin{figure}[!h]
\begin{adjustwidth}{-1in}{-1in}
\centering
\caption{Relationship Between Earnings and $p \left( X \right)$\label{fig:nswypx1}}
\includegraphics[width=12cm]{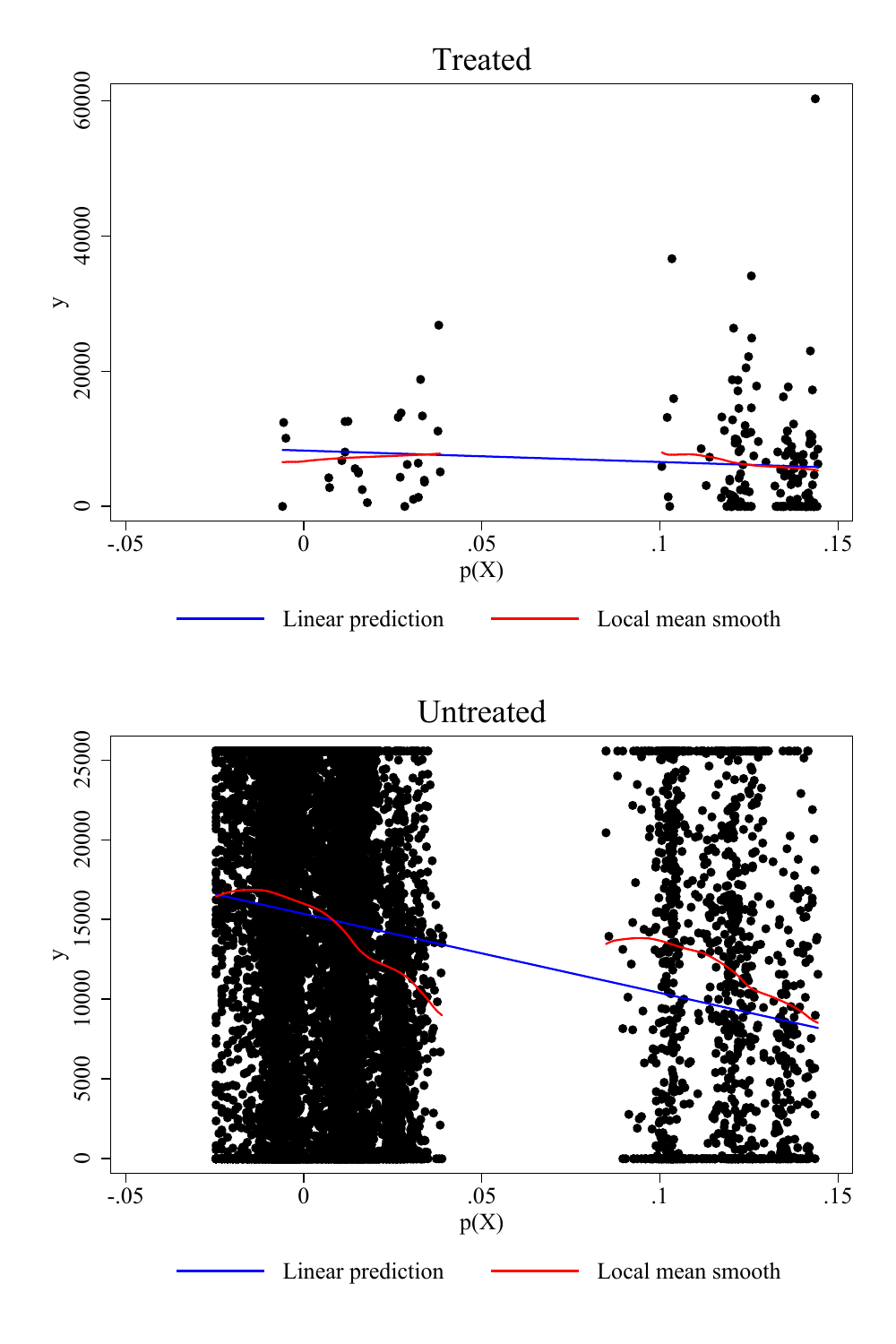}
\\
\vspace{-2mm}
\begin{footnotesize}
\begin{tabular}{p{12cm}}
\textit{Notes:} The vertical axis represents earnings in 1978. The horizontal axis represents the LPM propensity score. The propensity score is estimated using the specification in column 1 of Table \ref{tab:nsw1}. ``Local mean smooth'' is estimated using the Epanechnikov kernel and a rule-of-thumb bandwidth.
\end{tabular}
\end{footnotesize}
\end{adjustwidth}
\end{figure}

\begin{figure}[!h]
\begin{adjustwidth}{-1in}{-1in}
\centering
\caption{Relationship Between Earnings and $p \left( X \right)$\label{fig:nswypx2}}
\includegraphics[width=12cm]{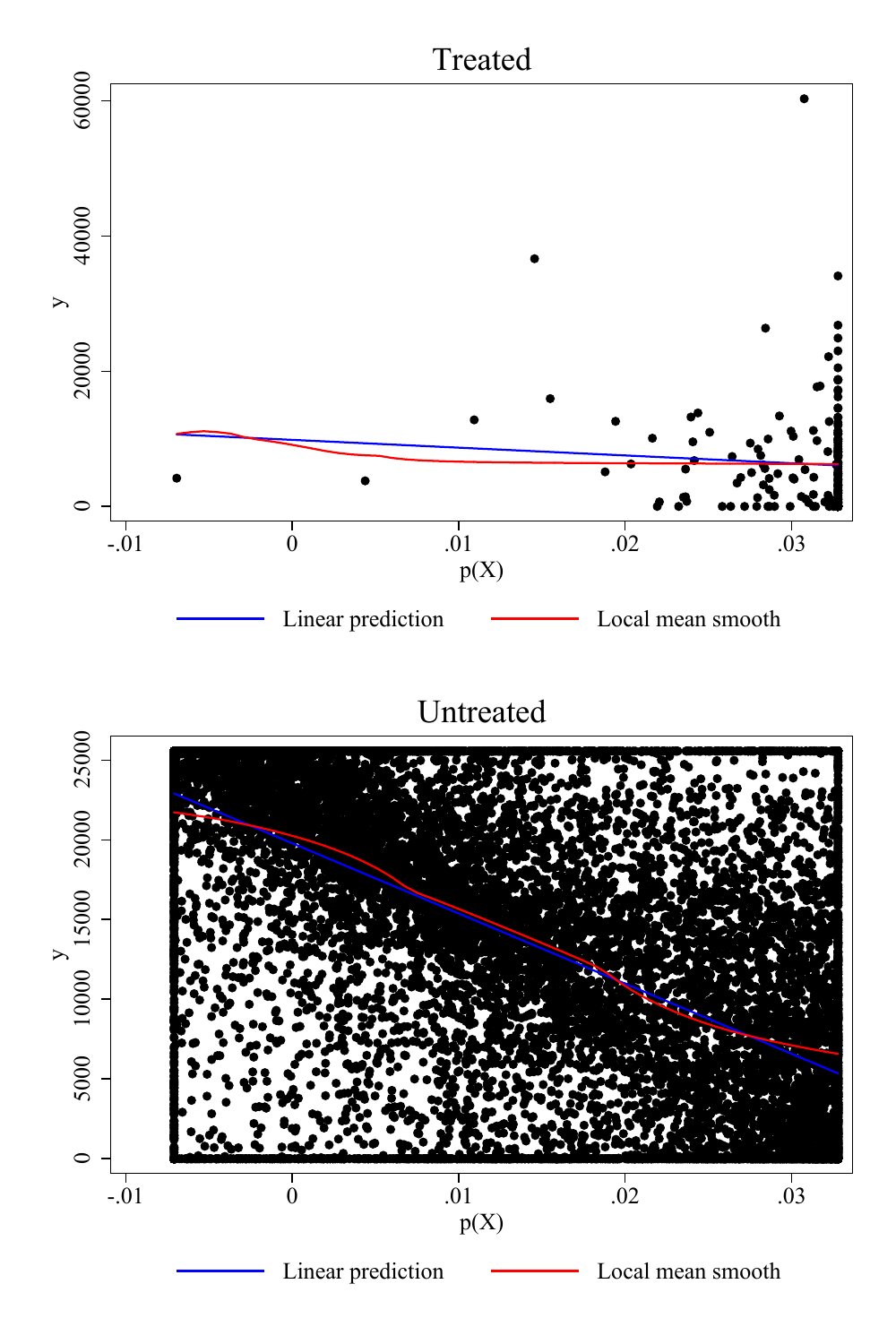}
\\
\vspace{-2mm}
\begin{footnotesize}
\begin{tabular}{p{12cm}}
\textit{Notes:} The vertical axis represents earnings in 1978. The horizontal axis represents the LPM propensity score. The propensity score is estimated using the specification in column 2 of Table \ref{tab:nsw1}. ``Local mean smooth'' is estimated using the Epanechnikov kernel and a rule-of-thumb bandwidth.
\end{tabular}
\end{footnotesize}
\end{adjustwidth}
\end{figure}

\begin{figure}[!h]
\begin{adjustwidth}{-1in}{-1in}
\centering
\caption{Relationship Between Earnings and $p \left( X \right)$\label{fig:nswypx3}}
\includegraphics[width=12cm]{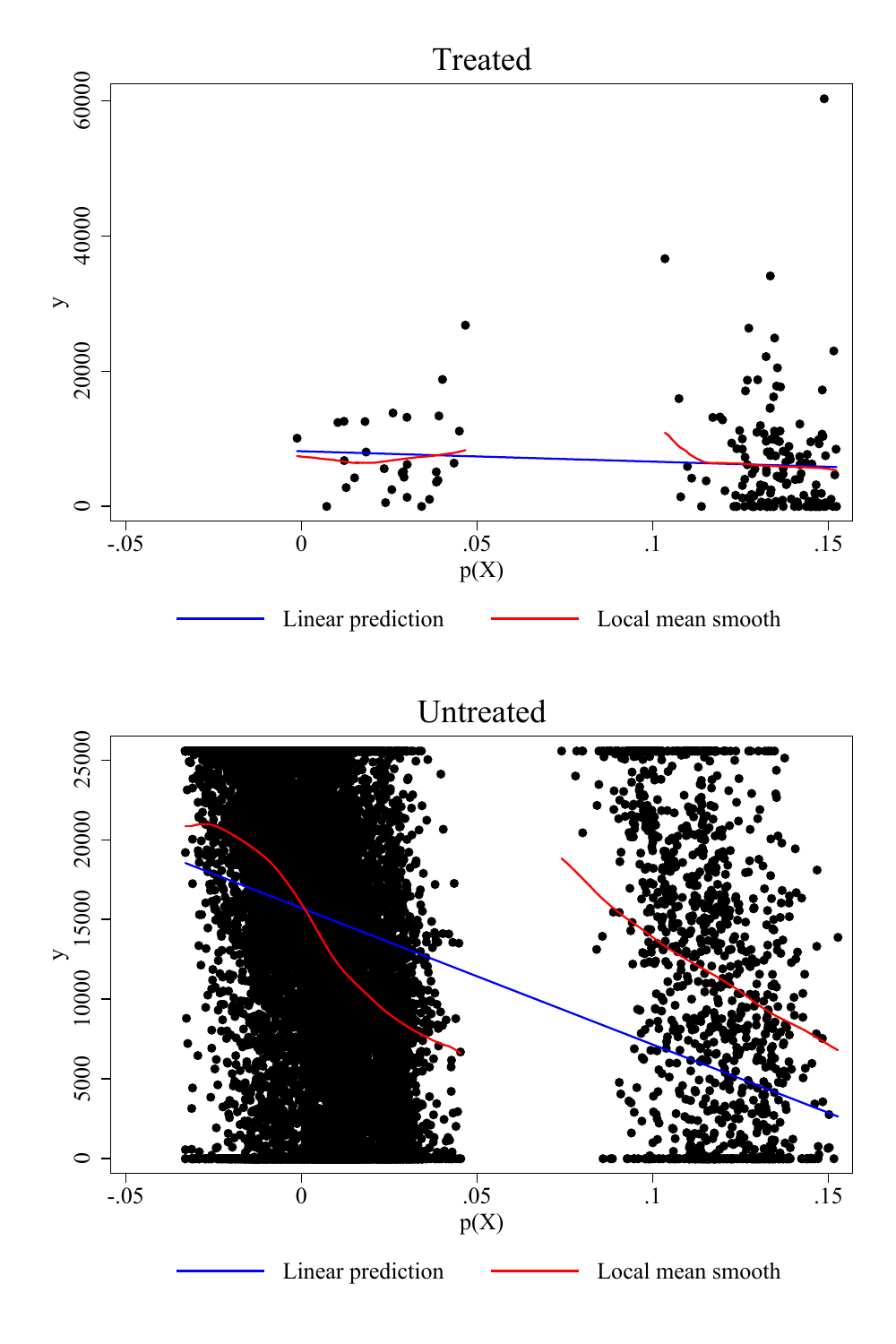}
\\
\vspace{-2mm}
\begin{footnotesize}
\begin{tabular}{p{12cm}}
\textit{Notes:} The vertical axis represents earnings in 1978. The horizontal axis represents the LPM propensity score. The propensity score is estimated using the specification in column 3 of Table \ref{tab:nsw1}. ``Local mean smooth'' is estimated using the Epanechnikov kernel and a rule-of-thumb bandwidth.
\end{tabular}
\end{footnotesize}
\end{adjustwidth}
\end{figure}

\begin{figure}[!h]
\begin{adjustwidth}{-1in}{-1in}
\centering
\caption{Relationship Between Earnings and $p \left( X \right)$\label{fig:nswypx4}}
\includegraphics[width=12cm]{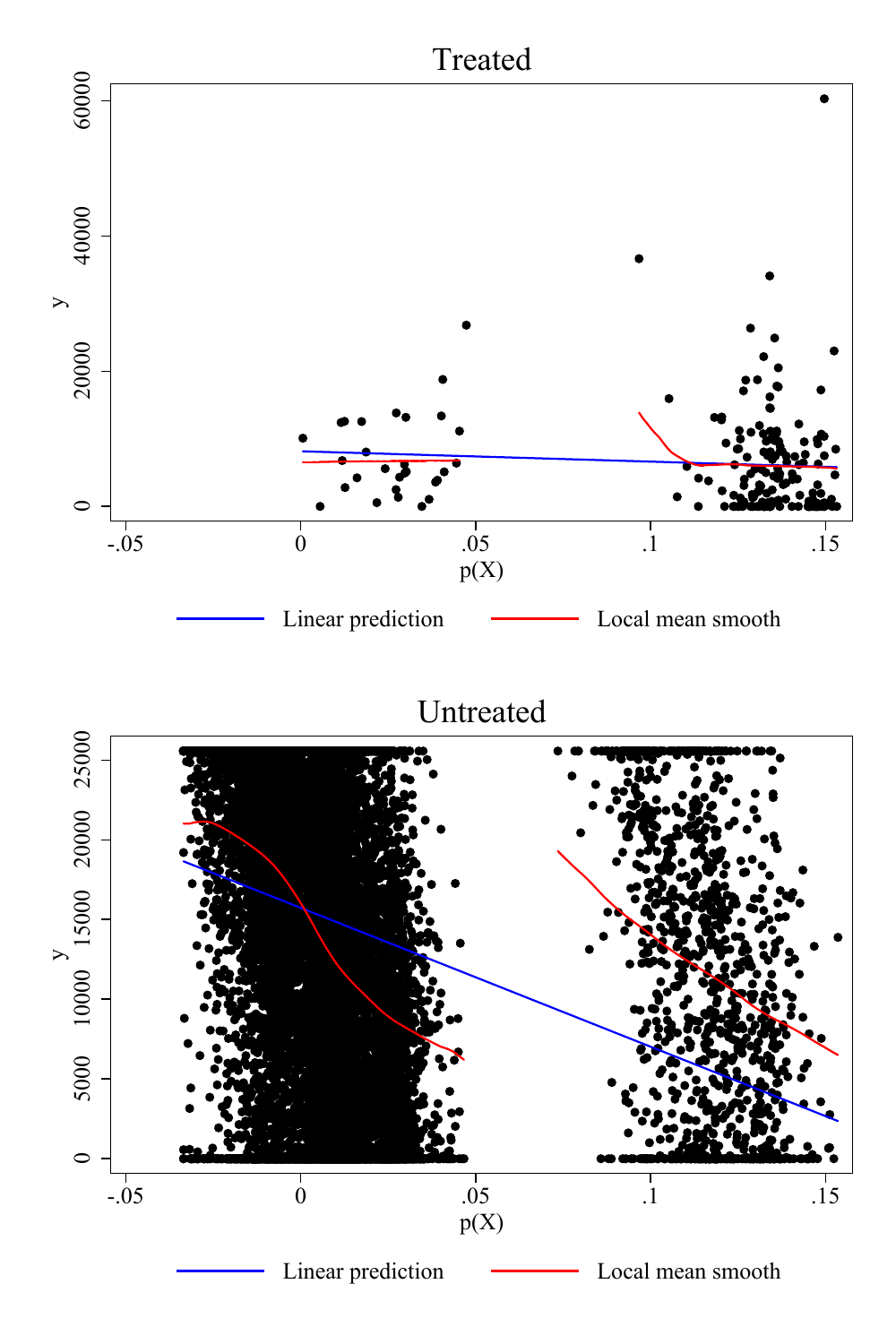}
\\
\vspace{-2mm}
\begin{footnotesize}
\begin{tabular}{p{12cm}}
\textit{Notes:} The vertical axis represents earnings in 1978. The horizontal axis represents the LPM propensity score. The propensity score is estimated using the specification in column 4 of Table \ref{tab:nsw1}. ``Local mean smooth'' is estimated using the Epanechnikov kernel and a rule-of-thumb bandwidth.
\end{tabular}
\end{footnotesize}
\end{adjustwidth}
\end{figure}

\begin{table}[!h]
\begin{adjustwidth}{-1in}{-1in}
\centering
\begin{threeparttable}
\begin{footnotesize}
\caption{Alternative Estimates of the Effects of a Training Program on Earnings\label{tab:nsw3}}
\begin{tabular}{l >{\centering\arraybackslash}m{2.5cm} >{\centering\arraybackslash}m{2.5cm} >{\centering\arraybackslash}m{2.5cm} >{\centering\arraybackslash}m{2.5cm}}
\hline\hline
    \multicolumn{1}{c}{} & (1) & (2) & (3) & (4) \\
\hline
    \multicolumn{1}{c}{} & \multicolumn{4}{c}{Matching on the LPM propensity score}  \\
\cline{2-5}
    $\widehat{\mathrm{ATE}}$   & --9,227*** & --7,504** & --6,245* & --6,581* \\
          & (2,388) & (3,518) & (3,382) & (3,370) \\
    $\widehat{\mathrm{ATT}}$   & --3,282*** & 257   & 975   & --892 \\
          & (863) & (694) & (813) & (906) \\
    $\widehat{\mathrm{ATU}}$   & --9,295*** & --7,594** & --6,328* & --6,646* \\
          & (2,415) & (3,556) & (3,420) & (3,409) \\
    \multicolumn{1}{l}{} & & & & \\
    \multicolumn{1}{c}{} & \multicolumn{4}{c}{Matching on the logit propensity score}  \\
\cline{2-5}
    $\widehat{\mathrm{ATE}}$   & --6,682** & --7,683*** & --4,187 & --2,961 \\
          & (2,773) & (2,421) & (3,012) & (11,900) \\
    $\widehat{\mathrm{ATT}}$   & --3,855*** & 265   & 2,117** & 2,032** \\
          & (854) & (695) & (856) & (860) \\
    $\widehat{\mathrm{ATU}}$   & --6,714** & --7,775*** & --4,260 & --3,018 \\
          & (2,804) & (2,448) & (3,046) & (12,037) \\
    \multicolumn{1}{l}{} & & & & \\
    \multicolumn{1}{c}{} & \multicolumn{4}{c}{Regression adjustment}  \\
\cline{2-5}
    $\widehat{\mathrm{ATE}}$   & --6,132*** & --6,218** & --4,952* & --4,930 \\
          & (1,644) & (2,534) & (2,996) & (3,073) \\
    $\widehat{\mathrm{ATT}}$   & --3,417*** & --69  & 623   & 796 \\
          & (628) & (598) & (628) & (639) \\
    $\widehat{\mathrm{ATU}}$   & --6,163*** & --6,289** & --5,017* & --4,996 \\
          & (1,662) & (2,561) & (3,030) & (3,108) \\
    \multicolumn{1}{l}{} & & & & \\
    \multicolumn{1}{l}{Demographic controls} & \checkmark & & \checkmark & \checkmark \\
    \multicolumn{1}{l}{Earnings in 1974} & & & & \checkmark \\
    \multicolumn{1}{l}{Earnings in 1975} & & \checkmark & \checkmark & \checkmark \\
    \multicolumn{1}{l}{} & & & & \\
    \multicolumn{1}{l}{$\hat{\rho} = \hat{\pr} \left( d=1 \right)$} & 0.011 & 0.011 & 0.011 & 0.011 \\
    \multicolumn{1}{l}{Observations} & 16,177 & 16,177 & 16,177 & 16,177 \\
\hline
\end{tabular}
\begin{tablenotes}[flushleft]
\item \textit{Notes:} The dependent variable is earnings in 1978. Demographic controls include age, age squared, years of schooling, and indicators for married, high school dropout, black, and Hispanic. For treated individuals, earnings in 1974 correspond to real earnings in months 13--24 prior to randomization, which overlaps with calendar year 1974 for a number of individuals. For ``matching on the LPM propensity score'' and ``matching on the logit propensity score,'' estimation is based on nearest-neighbor matching on the estimated propensity score (with a single match). The propensity score is estimated using a linear probability model (LPM) or a logit model. For ``regression adjustment,'' estimation is based on the estimator discussed in \cite{Kline2011}. Huber--White standard errors (regression adjustment) and Abadie--Imbens standard errors (matching) are in parentheses. Abadie--Imbens standard errors ignore that the propensity score is estimated.
\item *Statistically significant at the 10\% level; **at the 5\% level; ***at the 1\% level.
\end{tablenotes}
\end{footnotesize}
\end{threeparttable}
\end{adjustwidth}
\end{table}

\clearpage

\subsection{The Effects of Cash Transfers on Longevity}
\label{app:aizer}

\setcounter{table}{0}
\renewcommand{\thetable}{\ref{app:aizer}.\arabic{table}}
\setcounter{figure}{0}
\renewcommand{\thefigure}{\ref{app:aizer}.\arabic{figure}}

\vspace{4.425cm}

\begin{figure}[!h]
\begin{adjustwidth}{-1in}{-1in}
\centering
\caption{WLS Estimates of the Effects of Cash Transfers on Longevity\label{fig:mpwls1}}
\includegraphics[width=12cm]{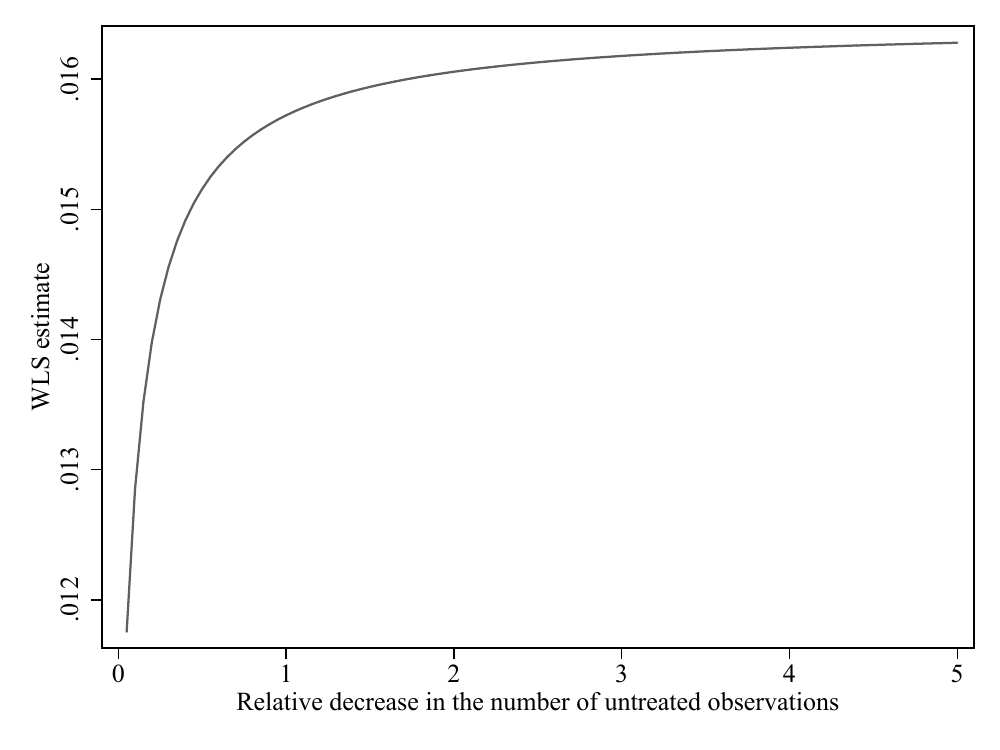}
\\
\vspace{-2mm}
\begin{footnotesize}
\begin{tabular}{p{12cm}}
\textit{Notes:} The vertical axis represents WLS estimates of the effect of cash transfers on log age at death using the model in equation (\ref{ols}) in the main text and the specification in column 1 of Table \ref{tab:mp1}, with weights of 1 for treated and $\frac{1}{k}$ for untreated units. The horizontal axis represents $k$.
\end{tabular}
\end{footnotesize}
\end{adjustwidth}
\end{figure}

\begin{figure}[!h]
\begin{adjustwidth}{-1in}{-1in}
\centering
\caption{WLS Estimates of the Effects of Cash Transfers on Longevity\label{fig:mpwls2}}
\includegraphics[width=12cm]{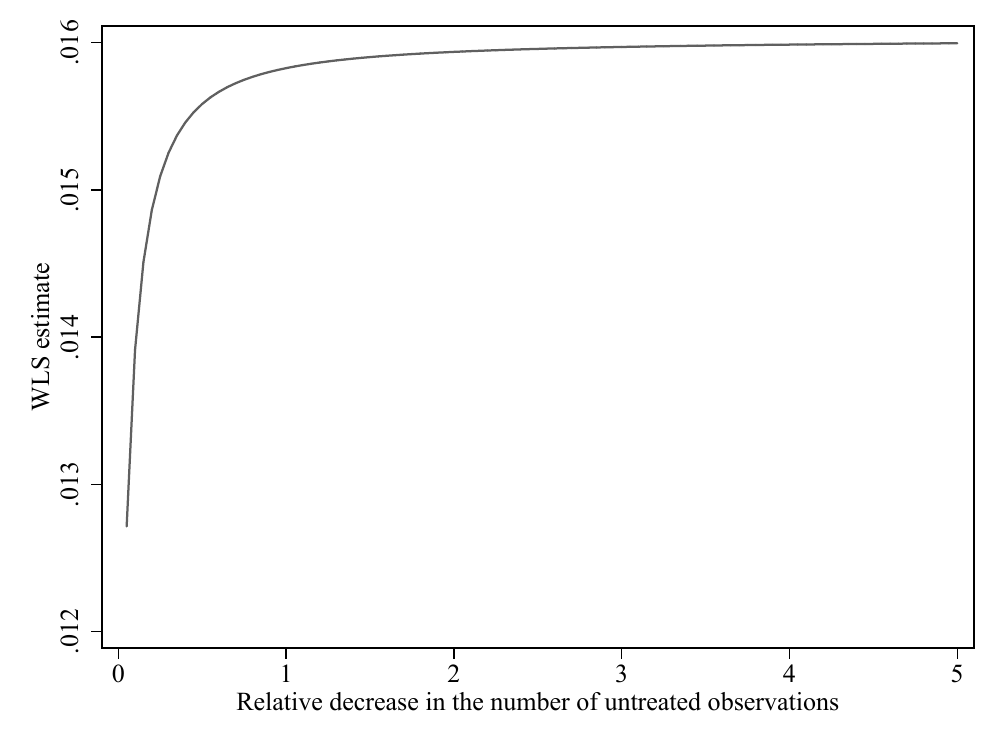}
\\
\vspace{-2mm}
\begin{footnotesize}
\begin{tabular}{p{12cm}}
\textit{Notes:} The vertical axis represents WLS estimates of the effect of cash transfers on log age at death using the model in equation (\ref{ols}) in the main text and the specification in column 2 of Table \ref{tab:mp1}, with weights of 1 for treated and $\frac{1}{k}$ for untreated units. The horizontal axis represents $k$.
\end{tabular}
\end{footnotesize}
\end{adjustwidth}
\end{figure}

\begin{figure}[!h]
\begin{adjustwidth}{-1in}{-1in}
\centering
\caption{WLS Estimates of the Effects of Cash Transfers on Longevity\label{fig:mpwls3}}
\includegraphics[width=12cm]{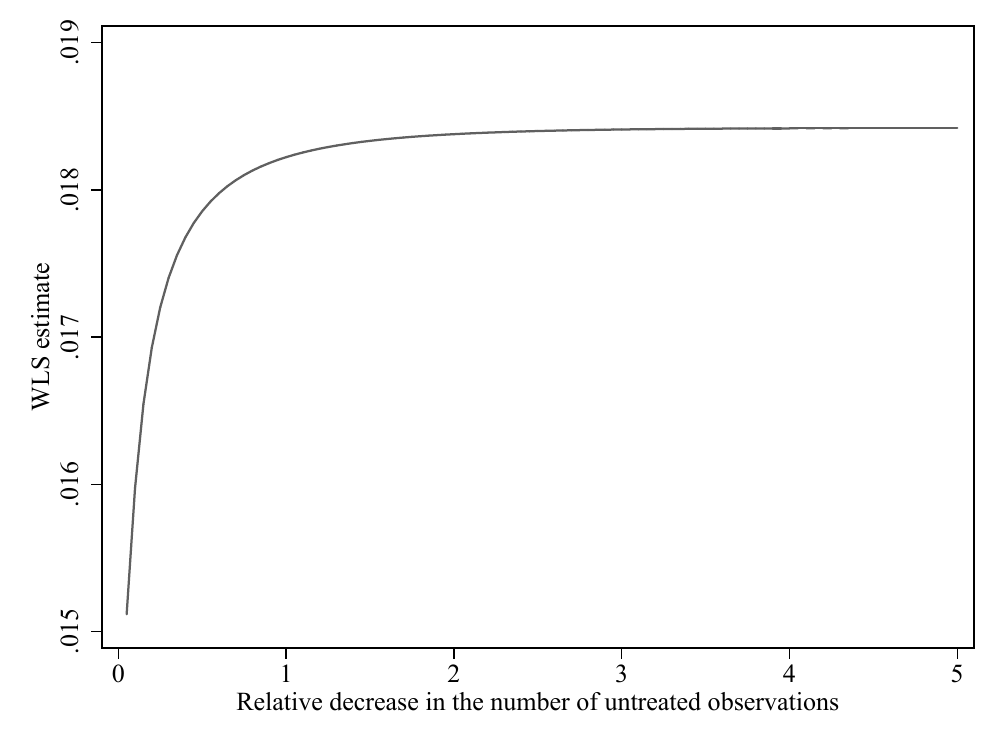}
\\
\vspace{-2mm}
\begin{footnotesize}
\begin{tabular}{p{12cm}}
\textit{Notes:} The vertical axis represents WLS estimates of the effect of cash transfers on log age at death using the model in equation (\ref{ols}) in the main text and the specification in column 3 of Table \ref{tab:mp1}, with weights of 1 for treated and $\frac{1}{k}$ for untreated units. The horizontal axis represents $k$.
\end{tabular}
\end{footnotesize}
\end{adjustwidth}
\end{figure}

\begin{figure}[!h]
\begin{adjustwidth}{-1in}{-1in}
\centering
\caption{WLS Estimates of the Effects of Cash Transfers on Longevity\label{fig:mpwls4}}
\includegraphics[width=12cm]{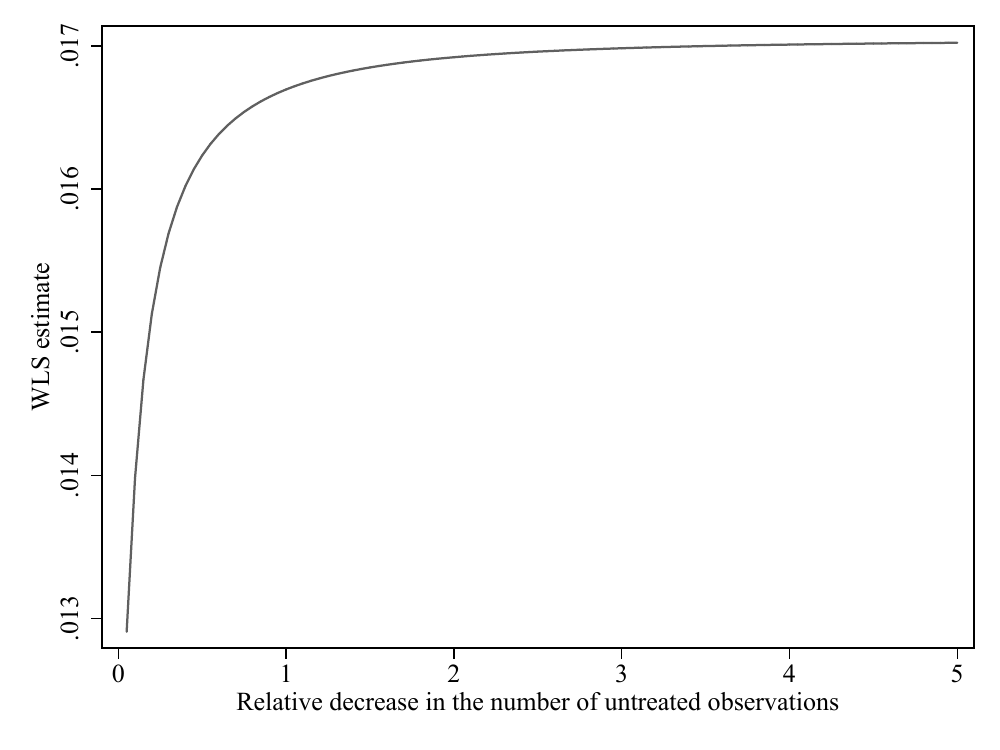}
\\
\vspace{-2mm}
\begin{footnotesize}
\begin{tabular}{p{12cm}}
\textit{Notes:} The vertical axis represents WLS estimates of the effect of cash transfers on log age at death using the model in equation (\ref{ols}) in the main text and the specification in column 4 of Table \ref{tab:mp1}, with weights of 1 for treated and $\frac{1}{k}$ for untreated units. The horizontal axis represents $k$.
\end{tabular}
\end{footnotesize}
\end{adjustwidth}
\end{figure}

\begin{figure}[!h]
\begin{adjustwidth}{-1in}{-1in}
\centering
\caption{Relationship Between Longevity and $p \left( X \right)$\label{fig:mpypx1}}
\includegraphics[width=12cm]{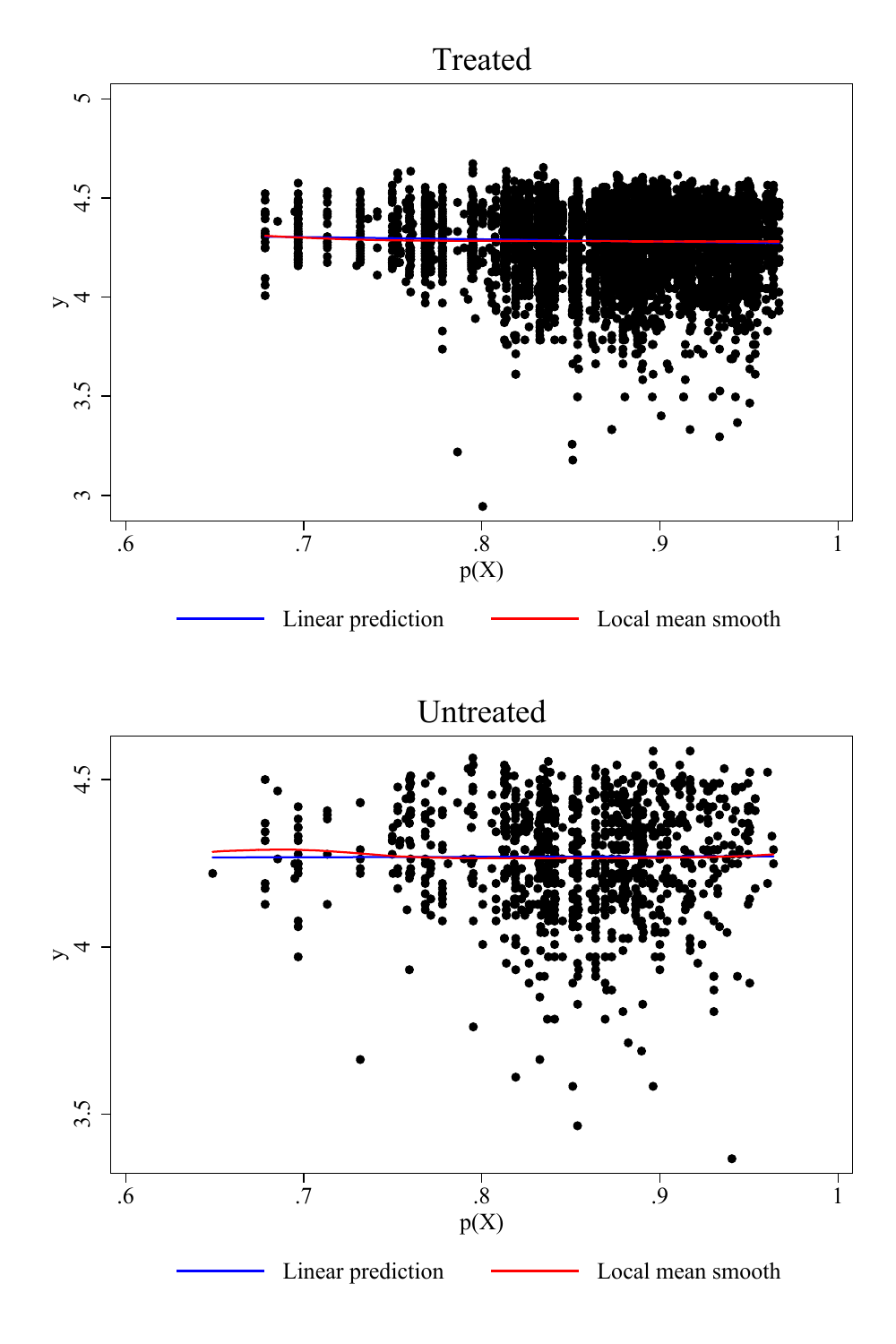}
\\
\vspace{-2mm}
\begin{footnotesize}
\begin{tabular}{p{12cm}}
\textit{Notes:} The vertical axis represents log age at death, as reported in the MP records. The horizontal axis represents the LPM propensity score. The propensity score is estimated using the specification in column 1 of Table \ref{tab:mp1}. ``Local mean smooth'' is estimated using the Epanechnikov kernel and a rule-of-thumb bandwidth.
\end{tabular}
\end{footnotesize}
\end{adjustwidth}
\end{figure}

\begin{figure}[!h]
\begin{adjustwidth}{-1in}{-1in}
\centering
\caption{Relationship Between Longevity and $p \left( X \right)$\label{fig:mpypx2}}
\includegraphics[width=12cm]{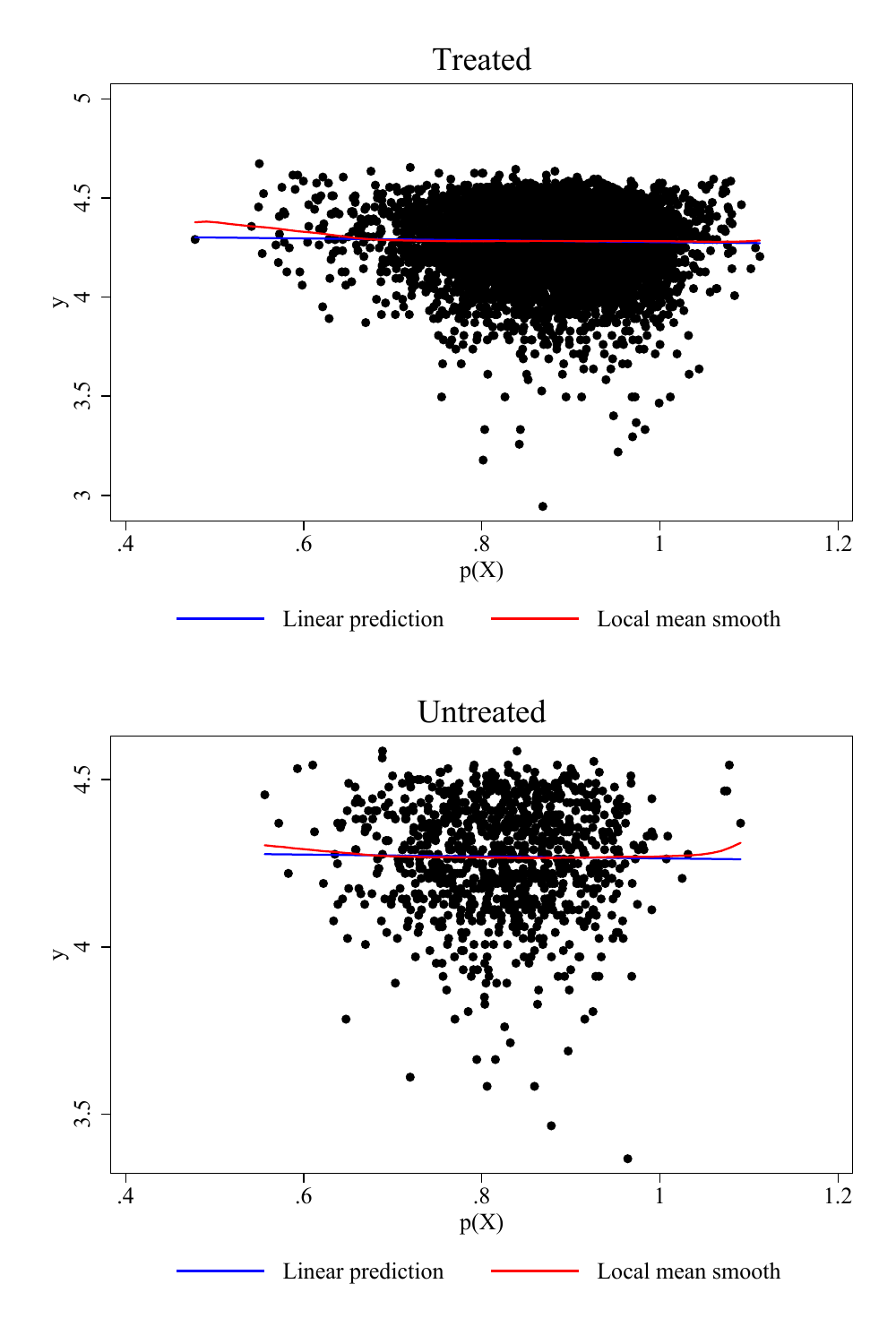}
\\
\vspace{-2mm}
\begin{footnotesize}
\begin{tabular}{p{12cm}}
\textit{Notes:} The vertical axis represents log age at death, as reported in the MP records. The horizontal axis represents the LPM propensity score. The propensity score is estimated using the specification in column 2 of Table \ref{tab:mp1}. ``Local mean smooth'' is estimated using the Epanechnikov kernel and a rule-of-thumb bandwidth.
\end{tabular}
\end{footnotesize}
\end{adjustwidth}
\end{figure}

\begin{figure}[!h]
\begin{adjustwidth}{-1in}{-1in}
\centering
\caption{Relationship Between Longevity and $p \left( X \right)$\label{fig:mpypx3}}
\includegraphics[width=12cm]{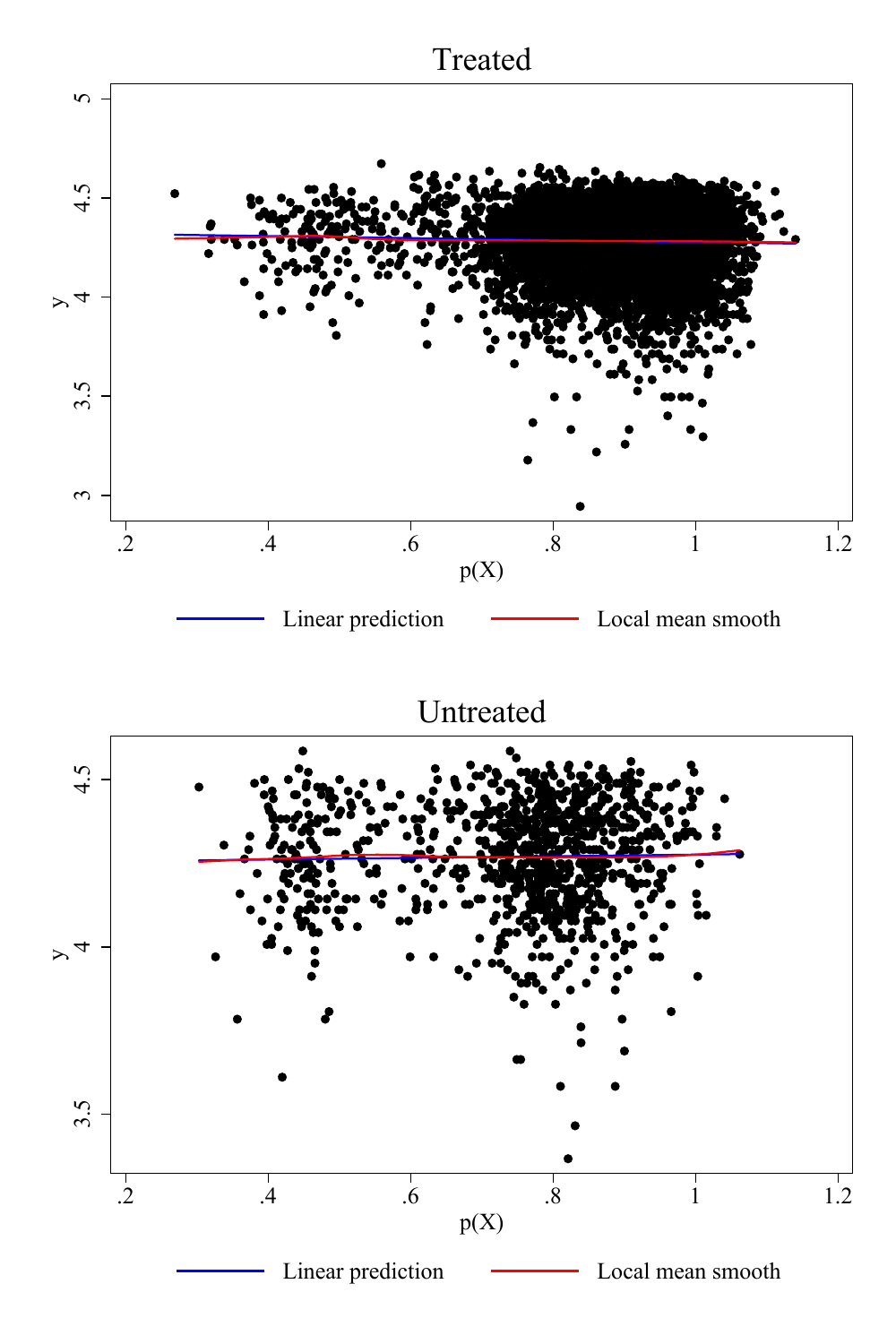}
\\
\vspace{-2mm}
\begin{footnotesize}
\begin{tabular}{p{12cm}}
\textit{Notes:} The vertical axis represents log age at death, as reported in the MP records. The horizontal axis represents the LPM propensity score. The propensity score is estimated using the specification in column 3 of Table \ref{tab:mp1}. ``Local mean smooth'' is estimated using the Epanechnikov kernel and a rule-of-thumb bandwidth.
\end{tabular}
\end{footnotesize}
\end{adjustwidth}
\end{figure}

\begin{figure}[!h]
\begin{adjustwidth}{-1in}{-1in}
\centering
\caption{Relationship Between Longevity and $p \left( X \right)$\label{fig:mpypx4}}
\includegraphics[width=12cm]{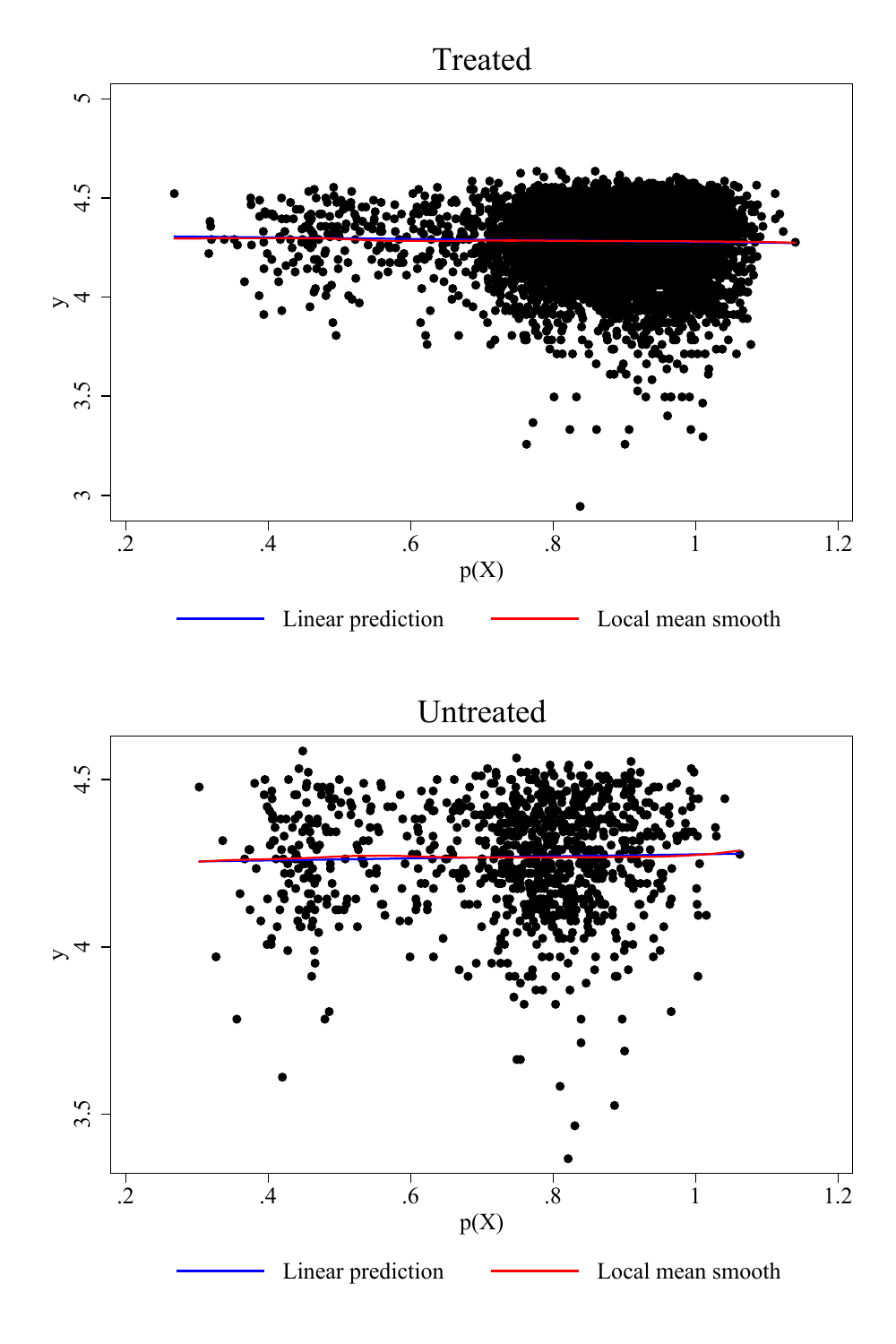}
\\
\vspace{-2mm}
\begin{footnotesize}
\begin{tabular}{p{12cm}}
\textit{Notes:} The vertical axis represents log age at death, as reported on the death certificate. The horizontal axis represents the LPM propensity score. The propensity score is estimated using the specification in column 4 of Table \ref{tab:mp1}. ``Local mean smooth'' is estimated using the Epanechnikov kernel and a rule-of-thumb bandwidth.
\end{tabular}
\end{footnotesize}
\end{adjustwidth}
\end{figure}

\begin{table}[!h]
\begin{adjustwidth}{-1in}{-1in}
\centering
\begin{threeparttable}
\begin{footnotesize}
\caption{Alternative Estimates of the Effects of Cash Transfers on Longevity\label{tab:mp3}}
\begin{tabular}{l >{\centering\arraybackslash}m{2.5cm} >{\centering\arraybackslash}m{2.5cm} >{\centering\arraybackslash}m{2.5cm} >{\centering\arraybackslash}m{2.5cm}}
\hline\hline
    \multicolumn{1}{c}{} & (1) & (2) & (3) & (4) \\
\hline
    \multicolumn{1}{c}{} & \multicolumn{4}{c}{Matching on the LPM propensity score}  \\
\cline{2-5}
    $\widehat{\mathrm{ATE}}$   & 0.0110 & 0.0147* & 0.0022 & 0.0011 \\
          & (0.0070) & (0.0089) & (0.0099) & (0.0098) \\
    $\widehat{\mathrm{ATT}}$   & 0.0106 & 0.0143 & --0.0002 & --0.0002 \\
          & (0.0073) & (0.0096) & (0.0109) & (0.0107) \\
    $\widehat{\mathrm{ATU}}$   & 0.0144** & 0.0179** & 0.0194** & 0.0100 \\
          & (0.0059) & (0.0082) & (0.0084) & (0.0085) \\
    \multicolumn{1}{l}{} & & & & \\
    \multicolumn{1}{c}{} & \multicolumn{4}{c}{Matching on the logit propensity score}  \\
\cline{2-5}
    $\widehat{\mathrm{ATE}}$   & 0.0111 & 0.0183** & --0.0019 & --0.0054 \\
          & (0.0073) & (0.0081) & (0.0166) & (0.0166) \\
    $\widehat{\mathrm{ATT}}$   & 0.0107 & 0.0181** & --0.0043 & --0.0105 \\
          & (0.0077) & (0.0087) & (0.0187) & (0.0186) \\
    $\widehat{\mathrm{ATU}}$   & 0.0145** & 0.0193** & 0.0152* & 0.0309*** \\
          & (0.0059) & (0.0083) & (0.0085) & (0.0083) \\
    \multicolumn{1}{l}{} & & & & \\
    \multicolumn{1}{c}{} & \multicolumn{4}{c}{Regression adjustment}  \\
\cline{2-5}
    $\widehat{\mathrm{ATE}}$   & 0.0105* & 0.0100 & 0.0140 & 0.0130 \\
          & (0.0063) & (0.0070) & (0.0110) & (0.0110) \\
    $\widehat{\mathrm{ATT}}$   & 0.0096 & 0.0092 & 0.0133 & 0.0124 \\
          & (0.0064) & (0.0073) & (0.0121) & (0.0121) \\
    $\widehat{\mathrm{ATU}}$   & 0.0164*** & 0.0160*** & 0.0184*** & 0.0170*** \\
          & (0.0058) & (0.0061) & (0.0065) & (0.0065) \\
    \multicolumn{1}{l}{} & & & & \\
    \multicolumn{1}{l}{State fixed effects} & \checkmark & & & \\
    \multicolumn{1}{l}{County fixed effects} & & & \checkmark & \checkmark \\
    \multicolumn{1}{l}{Cohort fixed effects} & \checkmark & \checkmark & \checkmark & \checkmark \\
    \multicolumn{1}{l}{State characteristics} & & \checkmark & \checkmark & \checkmark \\
    \multicolumn{1}{l}{County characteristics} & & \checkmark & & \\
    \multicolumn{1}{l}{Individual characteristics} & & \checkmark & \checkmark & \checkmark \\
    \multicolumn{1}{l}{} & & & & \\
    \multicolumn{1}{l}{$\hat{\rho} = \hat{\pr} \left( d=1 \right)$} & 0.875 & 0.875 & 0.875 & 0.875 \\
    \multicolumn{1}{l}{Observations} & 7,860 & 7,859 & 7,859 & 7,857 \\
\hline
\end{tabular}
\begin{tablenotes}[flushleft]
\item \textit{Notes:} The dependent variable is log age at death, as reported in the MP records (columns 1 to 3) or on the death certificate (column 4). State characteristics include manufacturing wages, age of school entry, minimum age for work permit, an indicator for a continuation school requirement, state laws concerning MP transfers (work requirement, reapplication requirement, and maximum amounts for first and second child), and log expenditures on education, charity, and social programs. County characteristics include average value of farm land, mean and SD of socio-economic index, poverty rate, female lfp rate, and shares of urban population, widowed women, children living with single mothers, and children working. Individual characteristics include child age at application, age of oldest and youngest child in family, number of letters in name, and indicators for the number of siblings, the marital status of the mother, and whether date of birth is incomplete. For ``matching on the LPM propensity score'' and ``matching on the logit propensity score,'' estimation is based on nearest-neighbor matching on the estimated propensity score (with a single match). The propensity score is estimated using a linear probability model (LPM) or a logit model. For ``regression adjustment,'' estimation is based on the estimator discussed in \cite{Kline2011}. Huber--White standard errors (regression adjustment) and Abadie--Imbens standard errors (matching) are in parentheses. Abadie--Imbens standard errors ignore that the propensity score is estimated.
\item *Statistically significant at the 10\% level; **at the 5\% level; ***at the 1\% level.
\end{tablenotes}
\end{footnotesize}
\end{threeparttable}
\end{adjustwidth}
\end{table}

\end{appendices}

\clearpage
\pagebreak

\setlength\bibsep{0pt}
\bibliographystyle{apalike}
\bibliography{Sloczynski_references}

\end{document}